\documentclass{article}

\usepackage{amsfonts}
\usepackage{amsmath}
\usepackage{amssymb}
\usepackage{amsthm}
\usepackage{array}
\usepackage{color}
\usepackage{comment}
\usepackage{enumerate}
\usepackage[mathscr]{euscript}
\usepackage{graphicx}
\usepackage[utf8]{inputenc}
\usepackage{latexsym}
\usepackage{mathrsfs}
\usepackage{mathtools}
\usepackage{appendix}
\usepackage{multirow}
\usepackage{subfigure}
\usepackage[a4paper, top=3cm, bottom=3.5cm, left=3.5cm, right=3.5cm]{geometry}
\usepackage{tikz}
\usepackage[font=small]{caption}
\usepackage[square,sort,comma,numbers]{natbib}
\RequirePackage[colorlinks,citecolor=blue,urlcolor=blue]{hyperref}
\usepackage[american]{babel}

\newcommand{\R}{\mathbb{R}}

\newcommand{\be}{\begin{equation}}
\newcommand{\ee}{\end{equation}}
\newcommand{\bea}{\begin{eqnarray}}
\newcommand{\eea}{\end{eqnarray}}
\newcommand{\beas}{\begin{eqnarray*}}
\newcommand{\eeas}{\end{eqnarray*}}
\newcommand{\ds}{\displaystyle}

\newtheorem{theorem}{Theorem}

\newtheorem{corollary}[theorem]{Corollary}

\newtheorem{definition}[theorem]{Definition}
\newtheorem{example}[theorem]{Example}

\newtheorem{lemma}[theorem]{Lemma}

\newtheorem{proposition}[theorem]{Proposition}
\newtheorem{remark}[theorem]{Remark}

\DeclareMathOperator*{\esssup}{ess\,sup}

\begin{document}

\title{\vskip -1.8cm Law-Invariant Return and Star-Shaped Risk Measures\thanks{This research was funded in part by the Netherlands Organization for Scientific Research under grant NWO Vici 2020--2027 (Laeven) and by an Ermenegildo Zegna Founder's Scholarship (Zullino). 
Emanuela Rosazza Gianin and Marco Zullino are members of Gruppo Nazionale per l’Analisi Matematica, la Probabilità e le loro Applicazioni (GNAMPA), Italy.}} 
\author{Roger J.~A.~Laeven\footnote{Corresponding author.} \\
{\footnotesize Dept.~of Quantitative Economics}\\
{\footnotesize University of Amsterdam, CentER}\\
{\footnotesize and EURANDOM, The Netherlands}\\
{\footnotesize \texttt{r.j.a.laeven@uva.nl}}\\
\and Emanuela Rosazza Gianin \\
{\footnotesize Dept.~of Statistics and Quantitative Methods}\\
{\footnotesize University of Milano Bicocca, Italy}\\
{\footnotesize \texttt{emanuela.rosazza1@unimib.it}}\\
\and Marco Zullino \\
{\footnotesize Dept.~of Mathematics and Applications}\\
{\footnotesize University of Milano Bicocca, Italy}\\
{\footnotesize \texttt{m.zullino@campus.unimib.it}}}

\date{This Version: \today }

\maketitle

\begin{abstract}

This paper presents novel characterization results for classes of law-invariant star-shaped functionals. 
We begin by establishing characterizations for positively homogeneous and star-shaped functionals that exhibit second- or convex-order stochastic dominance consistency.
Building on these characterizations, we proceed to derive Kusuoka-type representations for these functionals, shedding light on their mathematical structure and intimate connections to Value-at-Risk and Expected Shortfall. 
Furthermore, we offer representations of general law-invariant star-shaped functionals as robustifications of Value-at-Risk. 
Notably, our results are versatile, accommodating settings that may, or may not, involve monotonicity and/or cash-additivity.
All of these characterizations are developed within a general locally convex topological space of random variables, ensuring the broad applicability of our results in various financial, insurance and probabilistic contexts.\\[3mm]
\end{abstract}

\section{Introduction}

Over the past decades, a large literature has developed the theory of monetary risk measures --- monotone and cash-additive functionals --- and analyzed their applications in a variety of fields including economics, finance, insurance, operations research and statistics. 
More recently, \cite{BLR18} introduced return risk measures --- monotone and positively homogeneous functionals --- and first results in the development of their static and dynamic theory and applications were obtained in \cite{BLR18,BLR21,LR22,LRZ23,ABL23}.
Whereas monetary risk measures provide absolute assessments of risk, return risk measures provide relative assessments of risk, evocative of the distinct roles played by absolute and relative risk aversion measurements.
The positive homogeneity property of return risk measures was relaxed by \cite{LR22,LRZ23} asserting the more general star-shapedness property. 

Law-invariant risk measures --- also referred to as law-determined or distribution-invariant risk measures --- play an important role in the theory and applications of risk measures.\footnote{Law invariance is intimately related to probabilistic sophistication introduced by \cite{MS92}; see also \cite{M02,MMR06,S11,RS14}.}
This is due to their simplicity, tractability and statistical appeal, being statistical functionals.
Many well-known risk measures are specific examples of law-invariant risk measures (e.g., Value-at-Risk, Expected Shortfall, the entropic risk measure and the $p$-norm).
Law-invariant representations of coherent, convex, quasi-convex and quasi-logconvex risk measures have been derived in \cite{K01,FR05,CMMM11,LR22}. 

In this paper, we establish new characterization results for classes of law-invariant star-shaped functionals.
The existing literature offers general representations for law-invariant convex and quasi-convex functionals that do not require monotonicity and/or cash-additivity; see the recent \cite{BKMS21}.
For a comprehensive discussion of non-monotone preferences and their applications, see \cite{A07}.
There is, however, a notable gap in the literature for star-shaped functionals, which we aim to fill.
In \cite{LR22}, representations for law-invariant \emph{quasi-logconvex} star-shaped risk measures are derived.
Furthermore, in \cite{CCMTW22}, representations for law-invariant \emph{monetary} star-shaped risk measures are obtained, and in \cite{HWWX22} the latter results are extended to allow for cash-subadditivity. 
By contrast, our representation results for law-invariant star-shaped functionals do not rely on the properties of monotonicity, quasi-logconvexity and/or cash-(sub)additivity.
Furthermore, whereas \cite{LR22,CCMTW22,HWWX22} primarily focus on $L^{\infty}$ as the space of random variables, we derive our results in a general locally convex topological space of random variables. 

The generality and versatility of our setting --- both in terms of properties that our functionals may satisfy and in terms of the spaces they are defined on --- make some of the mathematical proofs intricate. 
Furthermore, while in the case of convex functionals the properties of law-invariance and convex-order stochastic dominance consistency are equivalent, provided that the functionals exhibit lower semicontinuity, such equivalence does not hold in general for star-shaped functionals, even not when assuming lower semicontinuity. 
Therefore, it becomes necessary to separately examine and analyze these two properties in the case of star-shaped functionals.

That is, the primary contribution of this paper is to bridge the gap between existing general representations for law-invariant convex and quasi-convex functionals on the one hand and those pertaining to star-shaped functionals on the other hand.
We first establish representations for second-order and convex-order stochastic dominance consistent (SSD- and CSD-consistent, for short) star-shaped functionals. 
We also showcase the robustness of our proof strategies: 
when we impose the additional axioms of monotonicity and cash-additivity, our results align with those in \cite{CCMTW22}, but within a (much) broader space than $L^{\infty}$. 
Second, we unveil novel representations \`a la Kusuoka \cite{K01} pertaining to SSD- and CSD-consistent star-shaped functionals, including their monetary variants that are previously unexplored in the literature.
These characterization results may be viewed as a contribution of independent interest.
Value-at-Risk and Expected Shortfall are pivotal building blocks in these representations.
Third, we extend our analysis to general law-invariant star-shaped functionals as robustifications of Value-at-Risk. 
We demonstrate the validity of these representations even without the need for monotonicity or cash-additivity. 
Furthermore, in the case of law-invariant star-shaped risk measures that are also cash-additive, we present a more precise representation compared to that in Theorem~5 of~\cite{CCMTW22}, resulting in a reduced set over which minimization takes place. 
We demonstrate that these three sets of results extend to positively homogeneous functionals, 
thus generalizing \cite{ABL23,LR22}, and yielding in particular new representation results for SSD-consistent return risk measures.
We finally provide three examples to illustrate that law-invariant star-shaped risk measures may arise naturally, also when cash-(sub)additivity is not preserved.

All the proof strategies share a common approach, relying on the approximation of star-shaped functions by convex ones. 
Consequently, our proof strategies effectively encapsulate the concept of star-shapedness, providing flexibility and adaptability in a myriad of settings in which star-shapedness is assumed.

The remainder of this paper is organized as follows.
In Section~\ref{sec:prel}, we provide some preliminaries for law-invariant star-shaped functionals.
In Section~\ref{sec:mr}, we establish our first representation results for 
SSD- and CSD-consistent star-shaped functionals.
Section~\ref{sec:kr} establishes representations \`a la Kusuoka \cite{K01} for SSD- and CSD-consistent star-shaped functionals.
In Section~\ref{sec:vr}, we derive representations as robustifications of Value-at-Risk for general law-invariant star-shaped functionals.
Finally, Section~\ref{sec:ex} provides three illustrative examples.

\section{Preliminaries}\label{sec:prel}

Let $(\Omega,\mathcal{F},\mathbb{P})$ be a non-atomic probability space and let $L^{0}$ be the space of real-valued random variables. 
Equalities and inequalities between random variables are meant to hold almost surely (a.s.). 
Given a set $\mathcal{X}\subseteq L^{0}$, we denote by $\mathcal{X}_+$ the subspace of $\mathcal{X}$ containing only positive random variables. 
The usual Lebesgue spaces of functions are denoted by $L^p$, $p\in[0,+\infty]$.

\begin{definition}
We say that $X\sim Y$ if $X$ and $Y$ have the same law under $\mathbb{P}$. 
Furthermore, a set $\mathcal{X}\subseteq L^0$ is said to be law-invariant if for any $X,Y\in L^0$ with $X\in\mathcal{X}$ and $X\sim Y$ it holds that $Y\in\mathcal{X}$. 
A functional $f:\mathcal{X}\to\R\cup\{+\infty\}$ is said to be law-invariant if for any $X,Y\in\mathcal{X}$ with $X\sim Y$ we have $f(X)=f(Y)$.
\end{definition}

We need the following assumptions on the spaces used in the sequel, which will be imposed throughout the paper unless explicitly mentioned otherwise. 
Let $\mathcal{X},\mathcal{X}^*$ be two linear subspaces of $L^0$ and such that: 

\begin{itemize}
    \item $\mathcal{X}$,$\mathcal{X}^*$ are law-invariant;
    \item for any $X\in\mathcal{X}$ and $Y\in\mathcal{X}^*$ it holds that $XY\in L^1$;
    \item $L^{\infty}\subseteq\mathcal{X}\subseteq L^1$ and $L^{\infty}\subseteq\mathcal{X}^*\subseteq L^1$;
    \item $\mathcal{X}$ contains the constants.
\end{itemize}
We consider the weakest topology $\sigma(\mathcal{X},\mathcal{X}^*)$ on $\mathcal{X}$ such that the linear functional $f_Y(X):=\mathbb{E}[XY]$ is continuous for every $Y\in\mathcal{X}^*$. 
Under these hypotheses, $\mathcal{X}$ is a locally convex topological space. 
Additionally, for any $X\in\mathcal{X}$ and $c\in\R$, it holds that $X+c\in\mathcal{X}$, because $\mathcal{X}$ is a linear space that includes constants. 
We will often use this fact without further mentioning.
We recall that Orlicz spaces and $L^p$ spaces with $p\in[1,+\infty]$ are examples of spaces satisfying the previous assumptions.
\begin{definition}\label{def:star}
    A functional $f:\mathcal{X}\to\R\cup\{+\infty\}$ with $f(0)<+\infty$ is star-shaped if for any $\lambda\in(0,1)$ and $X\in\mathcal{X}$ it holds that:
    $$f(\lambda X)\leq \lambda f(X)+(1-\lambda) f(0).$$ 
    Furthermore, a functional $\rho:\mathcal{X}\to\R\cup\{+\infty\}$ with $\rho(0)<+\infty$ is a risk measure if $\rho$ is increasing w.r.t.~the usual lattice order in $L^{0}$, i.e., if for any $X,Y\in\mathcal{X}$ such that $X\geq Y$ it holds that $\rho(X)\geq\rho(Y)$.
    We define the proper domain of $f$ as \mbox{$\mbox{Dom}(f):=\{X\in\mathcal{X}: f(X)<+\infty\}$}, and analogously for $\rho$.
\end{definition}
\noindent From this point onward, we will assume that any functional $f:\mathcal{X}\to\R\cup\{+\infty\}$ verifies $f(0)<+\infty$, i.e., $0\in\mbox{Dom}(f)$, unless specified otherwise.

We note that, different from a substantial part of the literature on risk measures, Definition~\ref{def:star} adopts the sign convention that risk measures satisfy increasing rather than decreasing monotonicity.
We recall that a functional $f:\mathcal{X}\to\R\cup\{+\infty\}$ is said to be lower semicontinuous if its level sets are $\sigma(\mathcal{X},\mathcal{X}^*)$-closed. 
In addition, if $f$ is law-invariant and convex, and $\mathcal{X}$ is a rearrangement invariant space, $\sigma(\mathcal{X},\mathcal{X}^*)$-lower semicontinuity is equivalent to the Fatou property, i.e., for any sequence $(X_n)_{n\in\mathbb{N}}\subseteq \mathcal{X}$ and $X\in\mathcal{X}$ s.t.~$X_n\to X$ a.s.~with $\ds\sup_{n\in\mathbb{N}}|X_n|\in\mathcal{X}$ it holds that $f(X)\leq\ds\liminf_{n\to\infty}{f(X_n)}$. 
If $f$ is monotone (thus, $f$ is a risk measure) these two properties are also equivalent to continuity from below, i.e., for any $(X_n)_{n\in\mathbb{N}}\subseteq \mathcal{X}$ and $X\in\mathcal{X}$ such that $X_n\uparrow X$ a.s.~we have $f(X)=\ds\lim_{n\to\infty}f(X_n)$. 

In the subsequent definition, we present a non-exhaustive set of axioms that a risk measure (or a functional) may satisfy.
\begin{definition}
    Let $\rho:\mathcal{X}\to\R\cup\{+\infty\}$ be a risk measure. 
    Then $\rho$ is said to be
    \begin{itemize}
    \item Normalized, if $\rho(0)=0$.
        \item Convex, if $\rho(\lambda X+(1-\lambda)Y)\leq \lambda\rho(X)+(1-\lambda)\rho(Y),$ for any $\lambda\in(0,1)$ and $X,Y\in\mathcal{X}$ with $X\neq Y$.
        \item Positively homogeneous, if $\rho(\lambda X)=\lambda\rho(X)$ for any $\lambda\in\R_+$ and $X\in\mathcal{X}$.
        \item Sublinear, if $\rho$ satisfies convexity and positive homogeneity.
        \item Cash-additive, if $\rho(X+m)=\rho(X)+m$ for any $m\in\R$ and $X\in\mathcal{X}$.
        \item Cash-subadditive, if $\rho(X+m)\leq\rho(X)+m$ for any $m\in\R_+$ and $X\in\mathcal{X}$.
    \end{itemize}
      A monetary risk measure $\rho:\mathcal{X}\to\R\cup\{+\infty\}$ is a risk measure that satisfies cash-additivity. 
      Similarly, a return risk measure $\rho:\mathcal{X}_+\to\R\cup\{+\infty\}$ is a risk measure verifying positive homogeneity.\footnote{In this paper, we consistently present our findings for risk measures defined across the entire space $\mathcal{X}$. 
      Nonetheless, these findings encompass the setting of return risk measures defined on $\mathcal{X}_+$ as a specific case, achieved by constraining the domain.}$^{,}$\footnote{Canonical examples of return risk measures are Orlicz premia --- Luxemburg norms from a purely mathematical perspective.
      Orlicz premia and their associated Orlicz spaces have been analyzed in a large literature in actuarial and financial mathematics (e.g., \cite{HG82,BF08,BF09,CL09,D12,LS14,BLR18}).}
    \label{def:axioms}
\end{definition}
    The axioms in Definition~\ref{def:axioms} have been extensively studied in the literature (see, e.g., \cite{GdVH84,FR02,FS02,FS11,D02,D06,ELKR09,LS13,DK13} and the references therein), also in connection with the law-invariance property (see, e.g., \cite{K01,FR05,BKMS21}). 

Let us briefly revisit the core stochastic dominance concepts, particularly first-, second- and convex-order stochastic dominance. 
Under first-order stochastic dominance, we say that, for given $X,Y\in\mathcal{X}$, $X$ dominates $Y$ if $\mathbb{E}[g(X)]\geq\mathbb{E}[g(Y)]$ for every increasing function $g:\mathbb{R}\to\mathbb{R}$ such that the expectations exist; it is denoted as $X\succeq_1Y$. 
Similarly, under second-order stochastic dominance, $X$ dominates $Y$ if $\mathbb{E}[g(X)]\geq\mathbb{E}[g(Y)]$ for every increasing and convex function $g:\mathbb{R}\to\mathbb{R}$; it is denoted as $X\succeq_2Y$. 
Finally, $X$ dominates $Y$ in the convex order sense if $\mathbb{E}[g(X)]\geq\mathbb{E}[g(Y)]$ for every (not necessarily increasing) convex function $g:\mathbb{R}\to\mathbb{R}$; it is denoted as $X\succeq_c Y.$
We write $\sim_{\cdot}$ when both $\succeq_{\cdot}$ and $\preceq_{\cdot}$ apply.

\begin{definition}
    A risk measure $\rho:\mathcal{X}\to\R\cup\{+\infty\}$ is:
    \begin{itemize}
        \item FSD-consistent, if $\rho(X)\geq\rho(Y)$ for any $X,Y\in\mathcal{X}$ such that $X\succeq_1 Y$.
        \item SSD-consistent, if $\rho(X)\geq\rho(Y)$ for any $X,Y\in\mathcal{X}$ such that $X\succeq_2 Y$.
        \end{itemize}
   A functional $f:\mathcal{X}\to\R\cup\{+\infty\}$ is CSD-consistent if $f(X)\geq f(Y)$ for any $X,Y\in\mathcal{X}$ such that $X\succeq_c Y$. 
\end{definition}

\begin{remark}
    We note that first-order stochastic dominance can be equivalently formulated as $$X\succeq_1 Y \iff VaR_{\beta}(X)\geq VaR_{\beta}(Y) \ \forall \beta\in[0,1],$$ 
    where $$VaR_{\beta}(X):=\inf\{x\in\R| \mathbb{P}(X\leq x)\geq \beta\}, \mbox{ for any } X\in\mathcal{X},$$ is the usual Value-at-Risk.
    Similarly, second-order stochastic dominance is equivalent to
    $$X\succeq_2 Y \iff ES_{\beta}(X)\geq ES_{\beta}(Y) \ \ \ \forall \beta\in[0,1),$$
    where, for any $X\in\mathcal{X}$,
    $$ES_{\beta}(X):=\frac{1}{1-\beta}\int_{\beta}^1VaR_{m}(X)dm, \mbox{ if } \beta\in[0,1), \ \mbox{ and } ES_1(X):=\esssup(X),$$
is the usual Expected Shortfall.\footnote{It is worth noting that the relation $ES_{\beta}(X)\geq ES_{\beta}(Y)$ for all $\beta\in[0,1)$ is equivalent to \mbox{$ES_{\beta}(X)=ES_{\beta}(Y)$} for all $\beta\in[0,1]$. 
Clearly, the second condition implies the first one. 
Moreover, for any $Z\in\mathcal{X}\subseteq L^1$, by Proposition 2.37 in~\cite{PR07} we have $ES_1(Z)=\ds\lim_{\beta\to 1}ES_{\beta}(Z)$. 
Thus, it holds that: $ES_1(X)=\ds\lim_{\beta\to1}ES_{\beta}(X)\geq\ds\lim_{\beta\to 1}ES_{\beta}(Y)=ES_1(Y)$. 
In the following, we will interchangeably use these two equivalent formulations of SSD.} 
It is clear by this characterization that first-order stochastic dominance implies second-order stochastic dominance, i.e., if $X\succeq_1 Y$ then the relationship $X\succeq_2 Y$ also holds. Furthermore, convex-order stochastic dominance is equivalent to $X\succeq_2 Y$ with the further constraint \mbox{$\mathbb{E}(X)=\mathbb{E}(Y),$} i.e., $ES_{\beta}(X)\geq ES_{\beta}(Y) \ \forall \beta\in[0,1)$ and $\mathbb{E}(X)=\mathbb{E}(Y).$
Finally, we underline that if a functional $f:\mathcal{X}\to\R\cup\{+\infty\}$ is consistent with one of these three stochastic dominance conditions, then $f$ is law-invariant. 
However, it is important to note that the converse implication does not hold, as we will explore in what follows.
See \cite{FS11,BKMS21} for further details on this subject.
\label{rem:LI}
\end{remark}

It is well-known that a functional $f:\mathcal{X}\to\R\cup\{+\infty\}$ that adheres to either FSD- or SSD-consistency also displays monotonicity with respect to the standard pointwise order relation in $L^0$. 
Therefore, FSD- or SSD-consistency are incompatible with functionals that lack monotonicity. 
In contrast to FSD- and SSD-consistency, the criterion of CSD-consistency does not entail monotonicity.

\section{SSD- and CSD-Consistent Star-Shaped Functionals}\label{sec:mr}

In this section, we explore the implications of the properties of SSD- and CSD-consistency for star-shaped risk measures within the framework of min-max representations. 
It has been shown in the literature (see e.g., \cite{CCMTW22}) that not all law-invariant star-shaped risk measures can be expressed as the minimum of law-invariant convex risk measures. 
This is exemplified by VaR, which falls into the category of star-shaped risk measures that cannot be represented as the minimum of law-invariant convex risk measures. 
The following theorem establishes the connection between SSD-consistent star-shaped risk measures and the minimum of SSD-consistent convex risk measures. 

\begin{theorem}
A risk measure $\rho:\mathcal{X}\to\R\cup\{+\infty\}$ is SSD-consistent and star-shaped if and only if there exist a set of indexes $\Gamma$ and a family of SSD-consistent convex risk measures $(\tilde \rho_{\gamma})_{\gamma\in\Gamma}$ with $\tilde \rho_{\gamma}:\mathcal{X}\to\R\cup\{+\infty\}$ such that $\tilde \rho_{\gamma}(0)=\rho(0)$ for all $\gamma\in\Gamma$ and
\begin{equation}
    \rho(X)=\min_{\gamma\in\Gamma}\tilde \rho_{\gamma}(X), \ X\in\mathcal{X}.
    \label{eq:SSrap}
\end{equation}
In addition, $\rho:\mathcal{X}\to\R\cup\{+\infty\}$ is a SSD-consistent and positively homogeneous risk measure if and only there exists a family of SSD-consistent and sublinear risk measures $(\tilde \rho_{\gamma})_{\gamma\in\Gamma}$ such that $\rho$ can be represented as in Equation~\eqref{eq:SSrap}.
\label{th:mainth}
\end{theorem}

\begin{proof}
We provide a detailed proof for the star-shaped case, followed by a brief sketch of the proof for positively homogeneous risk measures, since it is similar to the star-shaped case.
\smallskip
\begin{center} 
\textit{The star-shaped case}  
\end{center}
`Only if': We start by proving that there exists a family of law-invariant convex functionals of which the pointwise minimum is $\rho$.
Let us consider $\Gamma:=\mbox{Dom}(\rho)$ and the family of functionals $(\rho_Z)_{Z\in\Gamma}$, whose elements $\rho_Z:\mathcal{X}\to\R\cup\{+\infty\}$ are defined as:

\begin{equation*}
    \rho_Z(X)=
    \begin{cases}
        \alpha \rho(Z)+(1-\alpha)\rho(0) \ \ &\mbox{ if } \exists \alpha\in[0,1] \mbox{ s.t. } X\sim_2\alpha Z, \\
        +\infty &\mbox{ otherwise.}
    \end{cases}
\end{equation*}
\smallskip

\textit{$\rho_Z$ is well-defined:}  
Indeed, if there exists $\alpha\in[0,1]$ such that $X\sim_2\alpha Z$, then $\alpha$ is unique. 
To see this, let us assume that there exist $\alpha_1,\alpha_2\in[0,1]$ such that $X\sim_2\alpha_1Z\sim_2\alpha_2Z,$ i.e., $$ES_{\beta}(X)=ES_{\beta}(\alpha_1 Z)=ES_{\beta}(\alpha_2 Z) \ \ \forall \beta\in[0,1].$$ 
By positive homogeneity of the Expected Shortfall we have $(\alpha_1-\alpha_2)ES_{\beta}(Z)=0 \ \ \forall \beta\in[0,1].$ 

Now we prove that $ES_{\beta}(Z)=0$ for any $\beta\in[0,1]$ if and only if $Z\equiv0$. 
Clearly, by normalization, if $Z=0$ then $ES_{\beta}(0)=0 \ \forall \beta\in[0,1]$. 
Conversely, it is straightforward to prove that if $ES_{\beta}(Z)=0$ for all $\beta\in[0,1]$, then $m\mapsto VaR_m(Z)=0$ a.s.~$m\in [0,1]$. 
Indeed, choosing $\beta=1$, by definition of $ES_1(Z)$, we have $\esssup(Z)\leq 0$, which by monotonicity yields $VaR_m(Z)\leq0$ a.s.~$m\in[0,1]$ and the properties of the Lebesgue integral ensure $VaR_m(Z)=0$ a.s.~$m\in [0,1]$. 
Hence, by contradiction, let us suppose there exists a non-null measure set $A$ such that $Z<0$ on $A$, then using standard measure-theoretic arguments we can find a measurable set $B\subseteq A$ with $\mathbb{P}(B)>0$ and a constant $c<0$ such that $Z<c$ on $B$.\footnote{To prove this statement, let $B$ a non-null measure set such that $X<0$ on $B$. 
Setting $B_n:=\{\omega\in\Omega: X(\omega)\leq-\frac{1}{n}\}$, then $B=\bigcup_{n=1}^{\infty}B_n$. 
Thus, $B_n\subseteq B_{n+1}$ for all $n\in\mathbb{N}$ and $0<\mathbb{P}(B)=\mathbb{P}(\bigcup_{n=1}^{\infty}B_n)\leq\sum_{n=1}^{\infty}\mathbb{P}(B_n)$. 
Hence, there exists $\bar{n}\in\mathbb{N}$ such that $A:=\sum_{n=1}^{\bar{n}}B_n$ and $\mathbb{P}(A)>0$. 
Moreover, $X\leq -\frac{1}{\bar{n}}$ on $A$, taking $c:=-\frac{1}{\bar{n}}$ the thesis follows.} 
Defining $Z_1=Z$ on $\Omega\setminus B$ and $Z_1=c$ on $B$, it follows by monotonicity that $VaR_m(Z)\leq VaR_m(Z_1)$ for any $m\in[0,1]$. 
Moreover, $\mathbb{P}(Z_1\leq c)\geq \mathbb{P}(B)\geq m$ for any $m\in[0,\mathbb{P}(B)]$, yielding $$VaR_m(Z_1):=\inf\{x\in\R: \mathbb{P}(Z_1\leq x)\geq m\}\leq \inf\{x\in\R: \mathbb{P}(Z_1\leq x)\geq \mathbb{P}(B)\}\leq c <0,$$  for any $m\in[0,\mathbb{P}(B)]$, thus $VaR_m(Z)\leq VaR_m(Z_1)<0$ for all $m\in[0,\mathbb{P}(B)]$. 
The contradiction follows by recalling that $VaR_m(Z)=0$ a.s.~$m\in[0,1]$.  

So, if $Z\not\equiv0$, $\alpha_1=\alpha_2$. 
The case $Z\equiv0$ is obvious given that in this case the only possibility to have a finite value of $\rho$ is to consider $X\equiv0$. 
In particular, if $X\sim_2 Y$, then either $X\sim_2 Y\sim_2\alpha Z$ holds for some $\alpha\in[0,1]$ yielding $\rho_Z(X)=\alpha \rho(Z)+(1-\alpha)\rho(0)=\rho_Z(Y)$, or $X\sim_2 Y$ but $X\not\sim_2\alpha Z$ for any $\alpha\in[0,1]$, leading to $\rho_Z(X)=\rho_Z(Y)=+\infty$. 
\smallskip

\textit{Properties of $\rho_Z$:} 
We start by proving $\rho_Z(X)\geq \rho(X)$ for any $X\in\mathcal{X}$. 
To see this, let us assume there exists $\alpha\in[0,1]$ such that $X\sim_2\alpha Z$ (otherwise there is nothing to prove). 
Then, $\rho_Z(X)=\alpha \rho(Z)+(1-\alpha)\rho(0)\geq \rho(\alpha Z)= \rho(X)$, by the SSD-consistency and star-shapedness of $\rho$.
Furthermore, if $X\in\Gamma$, we can take $Z=X$, obtaining $\rho_X(X)=\rho(X)$. 
If $X\not\in\Gamma$, then $\rho_Z(X)=+\infty$ for any $Z\in\Gamma$. 
Indeed, assuming, by contradiction, that $\rho_Z(X)<+\infty$, there exists $\alpha\in[0,1]$ such that $\rho_Z(X)=\alpha\rho(Z)+(1-\alpha)\rho(0)\geq\rho(\alpha Z)=\rho(X)=+\infty,$ which leads to a contradiction. 
Thus, $\rho_Z(X)=\rho(X)=+\infty.$ 
In summary, we have established that $\rho(X)=\ds\min_{Z\in\Gamma}\rho_Z(X)$.

We prove now that $\rho_Z$ is law-invariant.  
Indeed, when $X\sim Y$, the law-invariance property of $ES_{\beta}(\cdot)$ for all $\beta\in[0,1]$ entails $ES_{\beta}(X)=ES_{\beta}(Y)$ for any $\beta\in[0,1]$. 
Applying the result established above, we can deduce $\rho_Z(X)=\rho_Z(Y)$.
\smallskip

We are now prepared to define the family of SSD-consistent and convex risk measures $(\rho_Z)_{Z\in\Gamma}$. 
For any $Z\in\Gamma$, we set $$\tilde\rho_Z(X):=\inf\{\rho_Z(Y): Y\succeq_2 X\}, \ X\in\mathcal{X}.$$

\textit{Properties of $\tilde\rho_Z$:} 
$\tilde\rho_Z$ is SSD-consistent.  
Indeed, for any $X_1,X_2\in\mathcal{X}$ with $X_1\succeq_2 X_2$ we have that if $Y\succeq_2 X_1$ then $Y\succeq_2X_2$, thus $\tilde\rho_Z(X_1)\geq\tilde\rho_Z(X_2).$ 
Hence, $\tilde\rho_Z$ is also monotone and law-invariant (see Remark~\ref{rem:LI}).
By definition of the infimum, it follows that $\tilde\rho_Z\leq \rho_Z$. 
Moreover, for each fixed $X\in\mathcal{X}$ and $Y\succeq_2 X$ it holds that $\rho_Z(Y)\geq \rho(Y)\geq \rho(X)$ (as shown in the first step of the proof), taking the infimum over $Y\succeq_2 X$ we obtain $\tilde\rho_Z(X)\geq \rho(X)$. 
Summing up, $\rho(X)\leq\tilde\rho_Z(X)\leq \rho_Z(X)$, hence $\rho(X)=\min_{Z\in\Gamma}\tilde\rho_Z(X)$. 
Now we prove convexity of $\tilde\rho_Z$. 
To see this, let us consider $X_1,X_2\in\mathcal{X}$ with $X_1\neq X_2$ and $\lambda\in(0,1)$.  
If $\tilde\rho_Z(X_1)=+\infty$ or $\tilde\rho_Z(X_2)=+\infty$ there is nothing to prove. 
Moreover, if $\tilde\rho_Z(\lambda X_1+(1-\lambda)X_2)=+\infty$ then either $\tilde\rho_Z(X_1)=+\infty$ or $\tilde\rho_Z(X_2)=+\infty$. 
Indeed, if $\tilde\rho_Z(X_i)<+\infty$ for $i=1,2$, then there exist $\alpha_i\in[0,1]$, $i=1,2$ such that $\tilde\rho_Z(X_i)=\alpha_i\rho(Z)$, $Y_i=\alpha_iZ\succeq_2 X_i$ and $Y_i$ attains the infimum\footnote{See Proposition~\ref{prop:Kusoka} for a detailed proof of this statement.} in the definition of $\tilde\rho_Z(X_i)$. 
This yields $(\lambda\alpha_1+(1-\lambda)\alpha_2)Z\succeq_2\lambda X_1+(1-\lambda)X_2$, by positive homogeneity and subadditivity of Expected Shortfall. 
Thus by definition of $\tilde\rho_Z$ we have $\tilde\rho_Z(\lambda X_1+(1-\lambda)X_2)\leq (\lambda\alpha_1+(1-\lambda)\alpha_2)\rho(Z)<+\infty$, hence a contradiction. 
So, we only need to check the case $\tilde\rho_Z(\lambda X_1+(1-\lambda)X_2),\tilde\rho_Z(X_1),\tilde\rho_Z(X_2)<+\infty$. 
In this case we find at least one random variable $Y_i\succeq_2 X_i$ such that $Y_i\sim_2\alpha_i Z$, $\alpha_i\in[0,1]$, $\rho_Z(Y_i)=\alpha_i\rho(Z)$ for $i=1,2$. 
Furthermore, it holds that 
\begin{align*}
   &ES_{\beta}(\lambda X_1+(1-\lambda)X_2) \leq \lambda ES_{\beta}(X_1)+(1-\lambda)ES_{\beta}(X_2) \\
   &\leq ES_{\beta}(\lambda  Y_1)+ES_{\beta}((1-\lambda) Y_2)=ES_{\beta}(\lambda \alpha_1Z+(1-\lambda)\alpha_2Z),
\end{align*}
for any $\beta\in[0,1]$. 
The first inequality follows from sublinearity of $ES_{\beta}$, the second inequality is due to the condition $Y_i\succeq_2 X_i$ for $i=1,2$, while the equality holds by positive homogeneity of $ES_{\beta}$, SSD-consistency of $ES_{\beta}$ and the relation $Y_i\sim_2\alpha_iZ$ for $i=1,2$. 
Thus, $\lambda \alpha_1Z+(1-\lambda)\alpha_2 Z\succeq_2\lambda X_1+(1-\lambda)X_2$ and
\begin{align*}
&\tilde\rho_Z(\lambda X_1+(1-\lambda)X_2)\leq \tilde\rho_Z(\lambda \alpha_1 Z+(1-\lambda)\alpha_2 Z)\\
&\leq\rho_Z(\lambda \alpha_1 Z+(1-\lambda)\alpha_2 Z)=\lambda \rho_Z(\alpha_1 Z)+(1-\lambda)\rho_Z(\alpha_2 Z) \\
&=\lambda \rho_Z(Y_1)+(1-\lambda)\rho_Z(Y_2),
\end{align*}
where the first inequality follows from SSD-consistency of $\tilde\rho_Z$, the second inequality follows from the relation $\tilde\rho_Z\leq \rho_Z$, the third equality is by definition of $\rho_Z$, while the definition of $\rho_Z$ together with the relation $Y_i\sim_2\alpha_i Z$ with $i=1,2$ yields the last equality. 
Taking the infimum over $Y_i\succeq_2 X_i$ such that $Y_i\sim_2\alpha_i Z$ with $\alpha_i\in[0,1]$ and $i=1,2$ we get the thesis. \smallskip

`If': We proceed to demonstrate the converse implication. 
Given a collection of SSD-consistent convex risk measures $(\rho_{\gamma})_{\gamma\in\Gamma}$ that share the same value at $0$ for all $\gamma\in\Gamma$, it is established that their pointwise minimum forms a star-shaped risk measure (refer to Lemma~7 in \cite{LRZ23}). 
Our focus now is to establish that $\rho(X):=\min_{\gamma\in\Gamma}\rho_{\gamma}(X)$ upholds SSD-consistency.
To establish this, let us assume $X\succeq_2 Y$. 
By the SSD-consistency of the family $(\rho_{\gamma})_{\gamma\in\Gamma}$, it follows that $\rho_{\gamma}(X)\geq\rho_{\gamma}(Y)$ for all $\gamma\in\Gamma$. 
This yields the result: $$\rho(X)=\min_{\gamma\in\Gamma}\rho_{\gamma}(X)\geq\min_{\gamma\in\Gamma}\rho_{\gamma}(Y)=\rho(Y).$$ \smallskip

\begin{center}
\textit{The positively homogeneous case}
\end{center}
To prove the statement regarding positively homogeneous risk measures, we consider the family of functionals $(\rho_Z)_{Z\in\Gamma}$ defined as:
\begin{equation*}
      \rho_Z(X)=
    \begin{cases}
        \alpha \rho(Z) \ &\mbox{ if } \exists \alpha\in[0,+\infty) \mbox{ s.t. } X\sim_2\alpha Z, \\
        +\infty &\mbox{ otherwise.}
    \end{cases}
\end{equation*}
By following the same reasoning as above, it can be proved that $\rho_Z$ is well-defined for all $Z\in\Gamma$. 
In addition, $\rho_Z(X)\geq\rho(X)$ and $\rho_X(X)=\rho(X)$. 
Moreover, we observe that \mbox{$\tilde{\rho}_Z:=\inf\{\rho_Z(Y):Y\succeq_2 X\}$} inherits the properties of $\rho_Z$ and it is SSD-consistent and sublinear. 
Thus, the family $(\tilde\rho_{\gamma})_{\gamma\in\Gamma}$ with $\Gamma:=\mbox{Dom}(\rho)$ fulfills the second thesis of the theorem. 
The converse implication is straightforward by observing that the pointwise minimum of sublinear risk measures is positively homogeneous.
\end{proof}
\begin{corollary}
        $f:\mathcal{X}\to\R\cup\{+\infty\}$ is a CSD-consistent and star-shaped functional if and only if there exist a set of indexes $\Gamma$ and a family $(\tilde f_{\gamma})_{\gamma\in\Gamma}$ of CSD-consistent and convex functionals $\tilde f_{\gamma}:\mathcal{X}\to\R\cup\{+\infty\}$ with $\tilde f_{\gamma}(0)=f(0)$ for all $\gamma\in\Gamma$, such that
        \begin{equation}
            f(X)=\min_{\gamma\in\Gamma}\tilde f_{\gamma}(X), \ X\in\mathcal{X}.
            \label{eq:SSrapNM}
        \end{equation} 
        In addition, $f:\mathcal{X}\to\R\cup\{+\infty\}$ is a CSD-consistent and positively homogeneous functional if and only there exists a family of CSD-consistent and sublinear functionals $(\tilde f_{\gamma})_{\gamma\in\Gamma}$ such that $f$ can be represented as in Equation~\eqref{eq:SSrapNM}.
    \label{cor:nonmon}
\end{corollary}

\begin{proof}
    This can be proved similarly as in the proof of Theorem~\ref{th:mainth}, recalling that the condition $\mathbb{E}(g(X))\geq\mathbb{E}(g(Y))$ for any convex (not necessarily increasing) function $g:\R\to\R$ is equivalent to the relation $ES_{\beta}(X)\geq ES_{\beta}(Y)$ for all $\beta\in[0,1]$ and $\mathbb{E}(X)=\mathbb{E}(Y),$ see Lemma 3 in~\cite{BKMS21}. 
    Define:
    \begin{equation*}
        f_Z(X):=
        \begin{cases}
            \alpha f(Z)+(1-\alpha)f(0) \ &\mbox{ if } \exists \alpha\in[0,1] \mbox{ s.t. } X\sim_c \alpha Z, \\
            +\infty &\mbox{ otherwise},
         \end{cases}
    \end{equation*}
     for any $Z\in\Gamma:=\mbox{Dom}(f)$. 
     Now, the proof follows verbatim from the proof of Theorem~\ref{th:mainth} considering $\tilde f_Z(X):=\inf\{f_Z(Y):Y\succeq_c X\}, \ X\in\mathcal{X}$.
\end{proof}

In the following proposition, we show that by requiring cash-additivity, we can obtain a result akin to Theorem 4 in \cite{CCMTW22}, but within the broader space $\mathcal{X}$. 
\begin{proposition}
A functional $\rho:\mathcal{X}\to\R\cup\{+\infty\}$ is an SSD-consistent, star-shaped and cash-additive risk measure if and only if there exist a set of indexes $\Gamma$ and a family $(\tilde\rho_{\gamma})_{\gamma\in\Gamma}$ of SSD-consistent, convex and cash-additive risk measures $\tilde\rho_{\gamma}:\mathcal{X}\to\R\cup\{+\infty\}$ such that $\tilde\rho_{\gamma}(0)=\rho(0)$ for any $\gamma\in\Gamma$ and
$$\rho(X)=\min_{\gamma\in\Gamma}\tilde\rho_{\gamma}(X), \ X\in\mathcal{X}.$$
Furthermore, $\rho$ is SSD-consistent, positively homogeneous and cash-additive if and only if each element of the family $(\tilde\rho_{\gamma})_{\gamma\in\Gamma}$ is SSD-consistent, sublinear and cash-additive.
\label{cor:CA}
\end{proposition}

\begin{proof}
        We establish the implication that if $\rho$ is SSD-consistent, star-shaped and cash-additive, then there exists a family of SSD-consistent, convex and cash-additive risk measures whose pointwise minimum is $\rho$. 
        The converse implication follows analogously to the proof of Theorem~\ref{th:mainth}. 
        We assume $\rho(0)=0$ for brevity.

        Let us define for any $Z\in \Gamma:=\mbox{Dom}(\rho)$ the functional:
                    \begin{equation*}
                        \rho_Z(X):=\begin{cases}
                            \alpha \rho(Z) + c \ \ &\mbox{ if } X \mbox{ is non-constant and there exist }\\&\quad \alpha\in(0,1], c\in\R \mbox{ s.t. } X\sim_2\alpha Z + c, \\
                            c \ \ &\mbox{ if there exists } c\in\R \mbox{ s.t. } X\sim_2c, \\
                            +\infty &\mbox{ otherwise.}
                        \end{cases}
                    \end{equation*}
               Let us observe that if $Z$ is constant, then the initial case in the definition of $\rho_Z$ cannot occur. 
               Indeed, when $Z$ is constant, $ES_{\beta}(X)=ES_{\beta}(Z+c)$ holds for all $\beta\in[0,1]$, and utilizing the cash-additivity of $ES_{\beta}$, we deduce $ES_{\beta}(X-Z-c)=0$ for all $\beta\in[0,1]$. 
               By invoking the first part of the proof of Theorem~\ref{th:mainth}, we conclude that $X=Z+c$ a.s., implying that $X$ must be constant.
               Consequently, either $X$ is constant and $\rho_Z(X)=c$, or $\rho_Z(X)=+\infty$. 
               Moreover, $\rho_Z$ is well-defined. 
               Indeed, if $Z$ is constant and there exists $c\in\R$ such that $X\sim_2c$, then $X=c$ a.s.\ and $c$ is unique. 
               If $Z$ is not constant, suppose by contradiction there exist $\alpha_1,\alpha_2\in[0,1]$ and $c_1,c_2\in\R$ such that $\alpha_1\neq\alpha_2$ and $c_1\neq c_2$. 
               Then $ES_{\beta}(\alpha_1Z+c_1)=ES_{\beta}(\alpha_2Z+c_2)$, for any $\beta\in[0,1]$. 
               By cash-additivity and positive homogeneity of $ES_{\beta}$ we have $ES_{\beta}(Z)=\frac{c_1-c_2}{\alpha_1-\alpha_2}$ for any $\beta\in[0,1]$. 
               Hence, the previous considerations yield $Z=\frac{c_1-c_2}{\alpha_1-\alpha_2}$ a.s., i.e., $Z$ is constant, which is a contradiction. 
               A similar conclusion holds if we suppose by contradiction $\alpha_1=\alpha_2$ and $c_1\neq c_2$ or $\alpha_1\neq\alpha_2$ and $c_1=c_2$. 
               Furthermore, the cash-additivity of $\rho_Z$ becomes evident. 
               In particular, for any $m\in\R$, if $\rho_Z(X)<+\infty$, then $\rho_Z(X+m)=\alpha \rho(Z) + c +m =\rho_Z(X)+m$. 
               This relation also holds when $\rho_Z(X)=+\infty$.
               Now, we proceed to demonstrate that $\rho_Z(X)\geq \rho(X)$ for all $X\in\mathcal{X}$, and $\rho_X(X)=\rho(X)$ for $X\in \Gamma$. 
               Let us consider a non-constant $X\sim_2\alpha Z + c$ for some $\alpha\in(0,1]$ and $c\in\mathbb{R}$ (the case where $X=c$ is clear due to the cash-additivity of $\rho$). 
               In this setting, $\rho_Z(X)=\alpha \rho(Z) + c = \alpha \rho(Z+c/\alpha)\geq \rho(\alpha Z +c) = \rho(X)$, utilizing the cash-additivity and star-shapedness of $\rho$. 
               Moreover, for $X\in \Gamma$, we can set $Z=X$, yielding $\rho_X(X)=\rho(X)$. 
               Thus, for any $X\in\mathcal{X}$, we have $\rho(X)=\ds\min_{Z\in\Gamma}\rho_{Z}(X)$.
               We observe that $\rho_Z$ is law-invariant for any $Z\in\Gamma$. 
               To see this, let us consider $X\sim X'$, thus $X\sim_2\alpha Z +c$ if and only if $X'\sim_2\alpha Z +c$, by law-invariance of $ES_{\beta}$, which implies $\rho_Z(X)=\rho_Z(X')$.
                
               We can build a family of SSD-consistent, star-shaped and cash-additive risk measures $(\tilde\rho_Z)_{Z\in\Gamma}$ by introducing $\tilde{\rho}_Z(X):=\inf\{\rho_Z(Y):Y\succeq_2 X\}$.
               In this case, the property of cash-additivity remains intact, as shown by
               $$\tilde{\rho}_Z(X+m)=\inf\{\rho_Z(Y):Y\succeq_2 X+m\}=\inf\{\rho_Z(W+m):W\succeq_2 X\}=\tilde{\rho}_Z(X)+m,$$ where the last equality leverages the cash-additivity of $\rho_Z$ and the definition of $\tilde{\rho}_Z$. 
               Additionally, similarly to Theorem~\ref{th:mainth}, we establish $\rho\leq{\tilde\rho}_Z\leq \rho_Z$, and also convexity of $\tilde{\rho}_Z$ can be checked similarly. 
               As a result, the thesis is obtained.
               The thesis concerning positively homogeneous risk measures is obtained by considering the family of sublinear functionals $(\rho_Z)_{Z\in\Gamma}$ defined by:
                   \begin{equation*}
                        \rho_Z(X):=\begin{cases}
                            \alpha\rho(Z) + c \ \ &\mbox{ if } X \mbox{ is non-constant and there exist }\\&\quad \alpha\in(0,+\infty), c\in\R \mbox{ s.t. } X\sim_2\alpha Z + c, \\
                            c \ \ &\mbox{ if there exists } c\in\R \mbox{ s.t. } X\sim_2c, \\
                            +\infty &\mbox{ otherwise,}
                        \end{cases}
                    \end{equation*}
                    and the corresponding family of SSD-consistent functionals $(\tilde\rho_Z)_{Z\in\Gamma}$ given by: 
                    $$\tilde\rho_Z(X):=\inf\{\rho_Z(Y): \ Y\succeq_2 X\}.$$
\end{proof}    
The following corollary extends our finding to CSD-consistent and star-shaped functionals, which are not necessarily monotone.
        \begin{corollary}
        $f:\mathcal{X}\to\R\cup\{+\infty\}$ is a CSD-consistent, star-shaped and cash-additive functional if and only if there exist a set of indexes $\Gamma$ and a family $(f_{\gamma})_{\gamma\in\Gamma}$ of CSD-consistent, convex and cash-additive functionals $f_{\gamma}:\mathcal{X}\to\R\cup\{+\infty\}$ with $f_{\gamma}(0)=f(0)$ for all $\gamma\in\Gamma$, such that $$f(X)=\min_{\gamma\in\Gamma}f_{\gamma}(X), \ X\in\mathcal{X}.$$ 
      Furthermore, $f$ is a CSD-consistent, positively homogeneous and cash-additive functional if and only if each element of the family $(f_{\gamma})_{\gamma\in\Gamma}$ is CSD-consistent, sublinear and cash-additive.
\end{corollary}

\begin{proof}
    The proof follows verbatim from the proofs of Corollary~\ref{cor:nonmon} and Proposition~\ref{cor:CA}.
\end{proof}

\section{Representations \`a la Kusuoka} \label{sec:kr}

In this section, we first present quantile-based representations for star-shaped functionals that satisfy SSD- or CSD-consistency. 
That is, we extend the representations established in \cite{K01,FR05,CMMM11} to our setting. 
We assume that $\mathcal{X}$ possesses a rearrangement-invariant structure, i.e., $\mathcal{X}$ is a solid lattice with a law-invariant and complete lattice norm (for further details on rearrangement invariant space cf.\ \cite{CR71}). 
This assumption is satisfied by commonly used spaces such as Orlicz spaces, Orlicz hearts and Marcinkiewicz spaces. 
As emphasized in Section~\ref{sec:prel}, within this framework, the concepts of $\sigma(\mathcal{X},\mathcal{X}^*)$-lower semicontinuity and the Fatou property are equivalent for any law-invariant convex functional $f:\mathcal{X}\to\R\cup\{+\infty\}$. 
Additionally, when $f$ is also monotonic, continuity from below is equivalent to the previous two properties (see Proposition~2.5 in \cite{BKMS21} for further details). 
We are ready to state the main result of this section.

\begin{proposition}
    Let $\mathcal{X}$ be a rearrangement-invariant space. 
    A functional \mbox{$f:\mathcal{X}\to\R\cup\{+\infty\}$} is CSD-consistent and star-shaped if and only if it can be represented as:
    \begin{equation}
    f(X)=\min_{\gamma\in\Gamma}\tilde f_{\gamma}(X)=\min_{\gamma\in\Gamma}\sup_{Y\in L^{\infty}}\left\{\int_0^1VaR_{\beta}(X)VaR_{\beta}(Y)d\beta-\tilde f_{\gamma}^*(Y)\right\}, \ \ X\in\mathcal{X},
            \label{eq:SDD_RS}
    \end{equation}
   where $(\tilde f_{\gamma})_{\gamma\in\Gamma}$ is the family of functionals provided in Corollary~\ref{cor:nonmon}, and $\tilde f_{\gamma}^{*}$ is given by: 
    \begin{equation*}
        \tilde f_{\gamma}^*(Y)=\sup_{X\in\mathcal{X}}\left\{\int_0^1VaR_{\beta}(X)VaR_{\beta}(Y)d\beta-\tilde f_{\gamma}(X)\right\}, \ \ Y\in\mathcal{X}^*,\gamma\in\Gamma.
    \end{equation*}
    In addition, if $\rho:\mathcal{X}\to\R\cup\{+\infty\}$ is an SSD-consistent and star-shaped risk measure, then the supremum in Equation~\eqref{eq:SDD_RS} can be taken over $L^{\infty}_+$, obtaining:
    $$\rho(X)=\min_{\gamma\in\Gamma}\tilde \rho_{\gamma}(X)=\min_{\gamma\in\Gamma}\sup_{Y\in L_+^{\infty}}\left\{\int_0^1VaR_{\beta}(X)VaR_{\beta}(Y)d\beta-\tilde \rho_{\gamma}^*(Y)\right\}, \ \ X\in\mathcal{X}.$$
    Here, $(\tilde \rho_{\gamma})_{\gamma\in\Gamma}$ is the family of risk measures provided in Theorem~\ref{th:mainth}, and 
     \begin{equation*}
        \tilde \rho_{\gamma}^*(Y)=\sup_{X\in\mathcal{X}}\left\{\int_0^1VaR_{\beta}(X)VaR_{\beta}(Y)d\beta-\tilde \rho_{\gamma}(X)\right\}, \ \ Y\in\mathcal{X}^*,\gamma\in\Gamma.
    \end{equation*}
 
    \noindent Finally, \mbox{$f:\mathcal{X}\to\R\cup\{+\infty\}$} is a CSD-consistent and positively homogeneous functional if and only if the following representation holds:
    $$f(X)=\min_{\gamma\in\Gamma}\tilde f_{\gamma}(X)=\min_{\gamma\in\Gamma}\sup_{Y\in \mathcal{M}^{\gamma}}\left\{\int_0^1VaR_{\beta}(X)VaR_{\beta}(Y)d\beta\right\}, \ \ X\in\mathcal{X},$$
    where $\Gamma$ is a set of indexes and $\mathcal{M}^{\gamma}\subseteq L^{\infty}$ for each $\gamma\in\Gamma$.  
    If $f$ is also monotone (thus, it is a risk measure), then the family of sets $(\mathcal{M}^{\gamma})_{\gamma\in\Gamma}$ can be chosen such that $\mathcal{M}^{\gamma}\subseteq L^{\infty}_+$.
    \label{prop:Kusoka}
\end{proposition}

\begin{proof} The `if' part is straightforward. 
Let us prove the `only if' implication. 
\begin{center}
    \textit{The CSD-consistent and star-shaped case}
\end{center}
According to Proposition~5.1 in~\cite{BKMS21}, if $\tilde f_{\gamma}$ is law-invariant, convex and $\sigma(\mathcal{X},\mathcal{X}^*)$-lower semicontinuous, it admits the representation: $$\tilde f_{\gamma}(X)=\sup_{Y\in L^{\infty}}\left\{\int_0^1VaR_{\beta}(X)VaR_{\beta}(Y)d\beta-\tilde f_{\gamma}^*(Y)\right\}, \ \ X\in\mathcal{X},$$ valid for any $\gamma\in\Gamma$. 
Here, $\tilde f_{\gamma}^*(Y)$ is given by:
$$\tilde f_{\gamma}^*(Y)=\sup_{X\in\mathcal{X}}\left\{\int_0^1VaR_{\beta}(X)VaR_{\beta}(Y)d\beta-\tilde f_{\gamma}(X)\right\}, \ \ Y\in\mathcal{X}^*.$$
Thus, it suffices to establish that $\tilde f_{\gamma}$ is $\sigma(\mathcal{X},\mathcal{X}^*)$-lower semicontinuous for any $\gamma\in\Gamma$, given that CSD-consistency and convexity of $\tilde f_{\gamma}$ have already been shown in Corollary~\ref{cor:nonmon}, recalling that CSD-consistency implies law-invariance (see Remark~\ref{rem:LI}). \smallskip

\textit{Equivalent expression for $\tilde f_Z$:} Referring back to Corollary~\ref{cor:nonmon} and considering the explicit form of the law-invariant convex functional $\tilde f_{\gamma}$, where $\gamma:=Z\in \mbox{Dom}(f)$, involved in the representation of $f$ (assuming for brevity that $f(0)=0$), we can express $\tilde f_Z$ as follows:
\begin{equation*}
   \tilde f_Z(X)=
    \begin{cases}
        \ds\inf_{\alpha\in{A}_X}\{\alpha f(Z)\} \ &\mbox{ if } A_{X}:=\{\alpha\in[0,1]: \exists Y\succeq_c X, Y\sim_c\alpha Z\}\neq\emptyset, \\
        +\infty &\mbox{ otherwise.}
    \end{cases}
\end{equation*}
We also recall the definition of $(f_Z)_{Z\in\Gamma}$:
\begin{equation*}
   f_Z(X)=
    \begin{cases}
        \alpha f(Z) \ &\mbox{ if } \exists\alpha\in[0,1] \mbox{ s.t. } X\sim_c \alpha Z, \\
        +\infty &\mbox{ otherwise.}
    \end{cases}
\end{equation*}
Note that for any fixed $X\in\mathcal{X}$ the infimum in the definition of $\tilde f_Z(X)$ is always attained as soon as \mbox{$A_{X}\neq \emptyset$.} 
Indeed, by definition of the infimum, there exist a sequence of real numbers $(\alpha_j)_{j\in\mathbb{N}}\subseteq[0,1]$ and a sequence of random variables $(Y_j)_{j\in\mathbb{N}}$ verifying $Y_j\sim_c \alpha_j Z$ and $Y_j\succeq_c X$ for all $j\in\mathbb{N}$ such that \mbox{$\ds\lim_{j\to\infty}\alpha_j f(Z)=\ds\inf_{\alpha\in A_X}\{\alpha f(Z)$\}.} We assume $f(Z)\neq0$, otherwise the statement is trivial, once observed that any $\alpha\in A_X$ is a minimizer in this case. 
In particular, we have that $$\ds\lim_{j\to\infty}\alpha_j=\frac{\ds\inf_{\alpha\in A_X}\{ \alpha f(Z)\}}{f(Z)}\in[0,1].$$ 
We want to prove that $\bar\alpha:=\ds\lim_{j\to\infty}\alpha_j$ is actually a minimizer. 
Let us consider the random variable $\bar Y:=\bar\alpha Z$. 
We need to check that $\bar Y\succeq_c X$ and $\bar Y \sim_c \bar\alpha Z$. 
The second condition is true by construction while the first condition follows from positive homogeneity of Expected Shortfall: $ES_{\beta}(Y_j)=\alpha_j ES_{\beta}(Z)\geq ES_{\beta}(X)$ for all $j\in\mathbb{N}$, and letting $j\to\infty$ we obtain $ES_{\beta}(\bar Y)=ES_{\beta}(\bar\alpha Z)=\ds\lim_{j\to\infty}\alpha_jES_{\beta}(Z)\geq ES_{\beta}(X)$ for all $\beta\in[0,1)$. 
In addition, it results that $$\ds\lim_{j\to\infty}\mathbb{E}(Y_j)=\ds\lim_{j\to\infty}\alpha_j\mathbb{E}(Z)=\mathbb{E}(\bar\alpha Z)=\mathbb{E}(Y)=\mathbb{E}(X),$$ thus the thesis follows. \smallskip

\textit{Lebesgue property of $ES_{\beta}$:} We claim that $ES_{\beta}$ possesses the Lebesgue property for any $\beta\in[0,1)$. 
Based on Proposition~2.35 in~\cite{PR07}, for any $X,Y\in L^1$ it follows that: $$|ES_{\beta}(X)-ES_{\beta}(Y)|\leq\frac{1}{1-\beta}\|X-Y\|_1.$$ 
Given the hypotheses, $\ds\sup_{n\in\mathbb{N}}|X_n|\in\mathcal{X}\subseteq L^1$, $X\in\mathcal{X}\subseteq L^1$, and $X_n\to X$ a.s., the dominated convergence theorem implies $X_n\to X$ in $L^1$. 
Consequently,
$$|ES_{\beta}(X_n)-ES_{\beta}(X)|\leq\frac{1}{1-\beta}\|X_n-X\|_1\to 0,$$
yielding $ES_{\beta}(X_n)\to ES_{\beta}(X)$ for any $\beta\in[0,1)$. \smallskip

\textit{$\sigma(\mathcal{X},\mathcal{X}^*)$-lower semicontinuity of $\tilde f_Z$:} As $\mathcal{X}$ is a rearrangement-invariant space, demonstrating the Fatou property of $\tilde f_{\gamma}$ for any $\gamma\in\Gamma$ is sufficient. 
This property requires showing that for a sequence $(X_n)_{n\in\mathbb{N}}\subseteq\mathcal{X}$ and $X\in\mathcal{X}$ with $X_n\to X$ a.s.\ and $\ds\sup_{n\in\mathbb{N}}|X_n|\in\mathcal{X}$, it holds that $\ds\liminf_{n\in\mathbb{N}}\tilde f_{\gamma}(X_n)\geq \tilde f_{\gamma}(X)$. Based on the preceding argument and assuming that $\tilde f_{Z}(X)\neq+\infty$, it is evident that  $\ds\liminf_{n\in\mathbb{N}}\tilde f(X_n)\neq +\infty$ only if there are infinitely many $n\in\mathbb{N}$ for which there exists $Y_n\succeq_c X_n$ with $Y_n\sim_c \alpha_n Z$, $\alpha_n\in[0,1]$ and $\tilde f(X_n)=f_Z(Y_n)=\alpha_nf(Z)$. 
This allows focusing on such sequences.
We need to verify that $\ds\liminf_{n\to\infty}\tilde f_{Z}(X_n)\geq \tilde f_{Z}(X)$. 
It holds that:
$$\liminf_{n\to\infty}\tilde f_{Z}(X_n)=\liminf_{n\to\infty}f_Z(Y_n)=\liminf_{n\to\infty}\alpha_n f(Z)=\tilde\alpha f(Z)=f_Z(\tilde Y),$$
where $\tilde Y:=\tilde\alpha Z$ with $\tilde\alpha:=\ds\liminf_{n\to\infty}\alpha_n$.
The last equality results from observing that $\tilde\alpha=\ds\liminf_{n\to\infty}\alpha_n\in[0,1]$ due to the fact that $\alpha_n\in[0,1]$ for all $n\in\mathbb{N}$, and in accordance with the definition of $f_Z$.
If we show that $\tilde Y\succeq_c X$ (thus,  $\tilde\alpha\in A_X$), then we can conclude $f_Z(\tilde Y)\geq\ds\inf_{\alpha\in A_X}{\alpha f(Z)}=\tilde f_Z(X)$, leading to  $\ds\liminf_{n\to\infty}\tilde f_Z(X_n)\geq \tilde f_Z(X)$. To establish this, we need to verify $ES_{\beta}(\tilde Y)\geq ES_{\beta}(X)$ and $\mathbb{E}(\tilde Y)=\mathbb{E}(X)$. We have:
\begin{align*}
&ES_{\beta}(\tilde Y)=\tilde\alpha ES_{\beta}(Z)=\liminf_{n\to\infty}\alpha_n ES_{\beta}(Z) \\
&=\liminf_{n\to\infty}ES_{\beta}(Y_n)\geq \liminf_{n\to\infty}ES_{\beta}(X_n)=ES_{\beta}(X),
\end{align*}
where the inequality is due to $Y_n\succeq_c X_n$ and the last equality holds by the Lebesgue property of $ES_{\beta}$. 
Thus, $ES_{\beta}(\tilde Y)\geq ES_{\beta}(X)$ for any $\beta\in[0,1)$. 
Analogously, by the dominated convergence theorem, we have $$\mathbb{E}(\tilde Y)=\liminf_{n\to\infty}\mathbb{E}(Y_n)=\liminf_{n\to\infty}\mathbb{E}(X_n)=\mathbb{E}(X).$$
By a similar argument involving contradiction, it can be established that when $\tilde f_{\gamma}(X)=+\infty$, $\ds\liminf_{n\to\infty}\tilde f(X_n)=+\infty$,  hence $\tilde f_Z$ verifies the Fatou property, yielding the $\sigma(\mathcal{X},\mathcal{X}^*)$-lower semicontinuity.

\begin{center}
    \textit{The SSD-consistent and star-shaped case}
\end{center}

Let us consider an SSD-consistent and star-shaped risk measure $\rho:\mathcal{X}\to\R\cup\{+\infty\}$. 
For brevity we suppose $\rho(0)=0$. 
Our goal is to establish the $\sigma(\mathcal{X},\mathcal{X}^*)$-lower semicontinuity of $\tilde\rho_\gamma$ for any $\gamma\in \Gamma$. 
Importantly, the SSD-consistency of $\rho$ inherently includes both monotonicity and law-invariance. 
Consequently, we can apply Proposition~5.1 as presented in \cite{BKMS21} once more to derive the desired dual representation for $\tilde \rho_{\gamma}$. \smallskip

\textit{Equivalent expression for $\tilde \rho_Z$:} By Theorem~\ref{th:mainth}, we know that $\tilde\rho_{\gamma}$ takes the following form, with $\gamma:=Z\in \mbox{Dom}(f)$:
$$\tilde\rho_Z(X):=\inf\{\rho_Z(Y): Y\succeq_2 X\},$$ which can be written as:
\begin{equation*}
    \tilde\rho_Z(X)=
    \begin{cases}
        \ds\inf_{\alpha\in A_{X}}\{\alpha \rho(Z)\} \ \ &\mbox{ if } A_{X}:=\{\alpha\in[0,1]: \exists Y\succeq_2 X, Y\sim_2\alpha Z\}\neq\emptyset, \\
        +\infty &\mbox{ otherwise,}
    \end{cases}
\end{equation*}
where 
\begin{equation*}
   \rho_Z(X)=
    \begin{cases}
        \alpha \rho(Z) \ &\mbox{ if } \exists\alpha\in[0,1] \mbox{ s.t. } X\sim_2 \alpha Z, \\
        +\infty &\mbox{ otherwise.}
    \end{cases}
\end{equation*}

Note that for any fixed $X\in\mathcal{X}$ the infimum in the definition of $\tilde\rho_Z(X)$ is always attained as soon as \mbox{$A_{X}\neq \emptyset$.} 
The proof of this statement is analogous to the case of CSD-consistency. \smallskip

\textit{$\sigma(\mathcal{X},\mathcal{X}^*)$-lower semicontinuity of $\tilde \rho_Z$:}
Under the monotonicity assumption, $\sigma(\mathcal{X},\mathcal{X}^*)$-lower semicontinuity is equivalent to continuity from below. 
Let $(X_n)_{n\in\mathbb{N}}\subseteq \mathcal{X}$ be a sequence such that $X_n\uparrow X\in\mathcal{X}$ a.s.; then we want to show that $\tilde\rho_{\gamma}(X_n)\to\tilde\rho_{\gamma}(X)$ for any $\gamma\in\Gamma$. 
By monotonicity, it is enough to prove that $\ds\liminf_{n\to\infty}\tilde\rho_{\gamma}(X_n)\geq\tilde\rho_{\gamma}(X)$ for any $\gamma\in\Gamma$. 
By hypotheses we have $X_n\uparrow X$, hence the continuity from below of Expected Shortfall ensures that $ES_{\beta}(X)=\ds\lim_{n\to\infty}ES_{\beta}(X_n)$. 
We start assuming $\tilde\rho_Z(X)<+\infty$. 
Once again, we can consider $(X_n)_{n\in\mathbb{N}}$ such that for each $n\in\mathbb{N}$ there exists $Y_n\succeq_2 X_n$ verifying $Y_n\sim_2 \alpha_nZ$ with $\alpha_n\in[0,1]$ attaining the minimum of the set $A_{X_n}$. 
Thus, 
$$ES_{\beta}(X)=\lim_{n\to\infty}ES_{\beta}(X_n)\leq\liminf_{n\to\infty} ES_{\beta}(Y_n)=\tilde\alpha ES_{\beta}(Z),$$
for any $\beta\in [0,1)$, with $\tilde\alpha:=\ds\liminf_{n\to\infty}\alpha_n$, where the inequality follows from the SSD-consistency of $ES_{\beta}$ and the relation $Y_n\succeq_2 X_n$ for all $n\in\mathbb{N}$. 
The above chain of inequalities, the definition of $\tilde\rho_Z$, the SSD-consistency of $\tilde\rho_Z$ and the relation $\tilde\rho_Z\leq\rho_Z$ yield:
$$\tilde\rho_Z(X)\leq \liminf_{n\to\infty} \tilde\rho_Z(Y_n)\leq \liminf_{n\to \infty}\rho_Z(\alpha_n Z)=\tilde\alpha \rho(Z)=\liminf_{n\to\infty}\tilde\rho_Z(X_n),$$
hence $\ds\liminf_{n\to\infty}\tilde\rho_Z(X_n)\geq \rho(X)$.
Similarly, if we assume $\rho_Z(X)=+\infty$ it can be proved that \mbox{$\ds\liminf_{n\to\infty}\rho_Z(X_n)=+\infty$}, thus the thesis follows.\smallskip

\textit{The positively homogeneous case:} In the case of positive homogeneity, we can similarly deduce the $\sigma(\mathcal{X},\mathcal{X}^*)$-lower semicontinuity of $\tilde\rho_Z$ and $\tilde f_Z$ for SSD- and CSD-consistent functionals, respectively. 
The key difference is that in this setting, the parameter $\alpha$ falls within the interval $\alpha\in[0,+\infty)$ rather than the previously considered $\alpha\in[0,1]$.
\end{proof}

The following corollary shows that we are also able to characterize cash-additive and SSD-consistent star-shaped risk measures through a representation à la Kusuoka, as robustification of Expected Shortfall. 
Notably, these findings appear to be unprecedented, even when considering the setting in which $\mathcal{X}=L^{\infty}$.

\begin{corollary}
Let $\mathcal{X}$ be a rearrangement-invariant space. 
A risk measure \mbox{$\rho:\mathcal{X}\to\R\cup\{+\infty\}$} is SSD-consistent, star-shaped and cash-additive if and only if $\rho$ admits the representation:
\begin{equation*}
    \rho(X)=\min_{\gamma\in\Gamma}\tilde\rho_{\gamma}(X)=\min_{\gamma\in\Gamma}\sup_{\mu\in\mathcal{P}((0,1])}\left\{\int_0^1ES_s(X)d\mu(s)-\alpha_{\gamma}(\mu)\right\},
\end{equation*}
where $\mathcal{P}((0,1])$ is the set of probability measures on $((0,1],\mathcal{B}((0,1])$, $\Gamma$ is a set of indexes, $(\tilde \rho_{\gamma})_{\gamma\in\Gamma}$ is the family of risk measures provided in Proposition~\ref{cor:CA} and $\alpha_{\gamma}$ is a penalty function defined for any $\mu\in\mathcal{P}((0,1])$ as:
\begin{equation*}
\alpha_{\gamma}(\mu):=\sup\left\{\int_{0}^{1}ES_s(X)d\mu(s); \ X\in L^{\infty}, \ \rho_{\gamma}(X)\leq0\right\}.
\end{equation*}
\end{corollary}

    \begin{proof}
              The `if' part is trivial. 
              Let us prove only the converse implication. 
              The proof is similar to the proof of Proposition~\ref{prop:Kusoka}. 
                For brevity we suppose $\rho(0)=0$ and we use the same notation as in Proposition~\ref{prop:Kusoka}.  We recall that for any $\gamma:=Z\in \mbox{Dom}(\rho)$ the functional $\tilde\rho_Z$ is given by $$\tilde\rho_Z(X):=\inf\{\rho_Z(Y):Y\succeq_2 X\}, \ X\in\mathcal{X}.$$ 
                We know that $\tilde\rho_Z$ is SSD-consistent, convex, monotone and cash-additive. 
                If we can prove that $\tilde\rho_Z$ is also continuous from below, then the thesis follows from Proposition~5.12 in~\cite{BKMS21}, as law-invariance is implied by SSD-consistency and $\sigma(\mathcal{X},\mathcal{X}^*)$-lower semicontinuity is equivalent to continuity from below for a monotone functional defined on a rearrangement invariant space.\smallskip
            
               \textit{Equivalent definition of $\tilde\rho_Z$:}
                Once again we can write an equivalent expression for $\tilde{\rho}_Z$. 
                Fixing $Z\in \mbox{Dom}(\rho)$, for each $X\in\mathcal{X}$ we define the set $B_X:=\{(\alpha,c)\in[0,1]\times\R:\exists Y\in\mathcal{X} \mbox{ s.t. } Y\succeq_2 X \mbox{ and } Y\sim_2\alpha Z+c\}$. 
                Then $\rho_Z$ takes the following form:
                \begin{equation*}
                   \tilde\rho_Z(X)=\begin{cases}
                        \ds\inf_{(\alpha,c)\in B_X}\{\alpha \rho(Z)+c\} \ \ &\mbox{ if } B_X\neq\emptyset, \\
            +\infty &\mbox{ otherwise. }
                    \end{cases}
                \end{equation*}
                We observe that when $B_X$ is non-empty the infimum in the previous equation must be finite given that $\tilde{\rho}_Z(X)\geq \rho(X)>-\infty$. 
                In addition, when $B_X$ is non-empty the infimum is in fact a minimum. 
                Indeed, arguing as in the proof of Proposition~\ref{prop:Kusoka}, there exist a sequence of real numbers $(\alpha_j,c_j)_{j\in\mathbb{N}}\subseteq[0,1]\times\R$ and a sequence of random variables $Y_j$ such that $\alpha_j \rho(Z)+c_j\to\ds\inf_{(\alpha,c)\in B_X}\{\alpha \rho(Z)+c\}$ for $j\to\infty$ and $Y_j\succeq X, Y_j\sim_2\alpha_j Z+c_j$ for any $j\in\mathbb{N}$. 
                By compactness of $[0,1]$ there exists a sub-sequence $(\alpha_{j_i})_{i\in\mathbb{N}}$ converging to some $\bar\alpha\in[0,1]$. 
                Moreover, $\ds\lim_{i\to\infty}\{\alpha_{j_i}\rho(Z)+c_{j_i}\}=\ds\inf_{(\alpha,c)\in B_X}\{\alpha \rho(Z)+c\}$, hence $\ds\lim_{i\to\infty}c_{j_i}$ exists and we define $\bar c:=\ds\lim_{i\to\infty}c_{j_i}$. Let us consider $\bar Y:=\bar \alpha Z+\bar c$. 
                By construction we have $\bar Y\sim_2 \bar\alpha Z+\bar c$. 
                Moreover, by cash-additivity and positive homogeneity of Expected Shortfall we have $\alpha_{j_i} ES_{\beta}(Z)+c_{j_i}\geq X$ for any $\beta\in[0,1)$ and $i\in\mathbb{N}$. 
                Letting $i\to\infty$ we get $ES_{\beta}(\bar Y)\geq X$ for any $\beta\in[0,1)$, hence the infimum is attained at $(\bar{\alpha},\bar c)\in[0,1]\times\R$.\smallskip
                
                \textit{Continuity from below of $\tilde\rho_Z$:} Given $X_n\uparrow X$ we want to prove that $\ds\lim_{n\to\infty}\tilde\rho_Z(X_n)=\tilde\rho_Z(X)$. 
                Due to monotonicity, it suffices to demonstrate $\ds\liminf_{n\to\infty}\tilde\rho_Z(X_n)\geq \tilde\rho_Z(X)$. We assume $\tilde\rho_Z(X)<+\infty$, which, by monotonicity, implies $\ds\liminf_{n\to\infty}\tilde\rho_Z(X_n)<+\infty$. 
                Thus, we can consider a sequence $(X_{n})_{n\in\mathbb{N}}$ such that $\tilde{\rho}_Z(X_{n})<+\infty \ \ \forall n\in\mathbb{N},$ verifying $\tilde{\rho}_Z(X_{n})={\rho}_Z(Y_{n})=\alpha_{n}\rho(Z)+c_{n}$. Here, for each $n\in\mathbb{N}$, $(\alpha_{n},c_{n})\in[0,1]\times\mathbb{R}$ attains the infimum of $B_{X_n}$, with $Y_{n}\succeq_2 X_{n}$ and $Y_{n}\sim_2\alpha_{n} Z+c_{n}$.
                Because $[0,1]$ is a compact set, we can extract a sub-sequence $(\alpha_{n_{k}})_{k\in\mathbb{N}}$ such that $\alpha_{n_{k}}\to\tilde{\alpha}\in[0,1]$. 
                By monotonicity of $\tilde{\rho}_Z$, we infer that $\tilde{\rho}_Z(X_{n})$ must converge to some $l\in\mathbb{R}$ (otherwise, we would have $\ds\liminf_{n\to\infty}\tilde\rho_Z(X_{n})=+\infty$, which contradicts the hypothesis). 
                Thus, also $\tilde{\rho}_Z(X_{n_{k}})\to l$ and then $c_{n_{k}}=\tilde\rho_Z(X_{n_{k}})-\alpha_{n_{k}}\rho(Z)\to l-\tilde{\alpha}\rho(Z):=\tilde{c}\in\R$.

Since $X_n\uparrow X$, continuity from below of Expected Shortfall ensures that $ES_{\beta}(X)=\ds\lim_{n\to\infty}ES_{\beta}(X_n)$. 
Thus, 
$$ES_{\beta}(X)=\lim_{k\to\infty}ES_{\beta}(X_{n_k})\leq\liminf_{k\to\infty} ES_{\beta}(Y_{n_k})=\tilde\alpha ES_{\beta}(Z)+\tilde c,$$
for any $\beta\in [0,1)$, where the inequality follows from the SSD-consistency of $ES_{\beta}$ and the relation $Y_{n_k}\succeq_2 X_{n_k}$ for all $k\in\mathbb{N}$. 
The above chain of inequalities, the definition of $\tilde\rho_Z$ and the SSD-consistency of $\tilde\rho_Z$ yield:
                \begin{eqnarray*}
                    \tilde{\rho}_Z(X)&\stackrel{\mathrm{SSD}}\leq & \lim_{k\to\infty}\tilde{\rho}_Z( Y_{n_{k}})\stackrel{\mathrm{\tilde{\rho}_Z\leq \rho_Z}}\leq \lim_{k\to\infty} \rho_Z(Y_{n_{k}}) \\ &\stackrel{\mathrm{def. \rho_Z}}=&\tilde{\alpha}\rho(Z)+\tilde{c}= \lim_{k\to\infty} \{\alpha_{n_k}\rho(Z)+c_{n_{k}}\} \\
&=&\lim_{k\to\infty}\tilde{\rho}_Z(X_{n_{k}})=\liminf_{n\to\infty}\tilde{\rho}_Z(X_n).
                \end{eqnarray*}
                If $\tilde{\rho}_Z(X)=+\infty$ we can proceed by contradiction, proving that also $\ds\liminf_{n\to\infty}\tilde{\rho}_Z(X_n)=+\infty$. 
                The thesis follows.                 
\end{proof}

\section{Law-Invariant Star-Shaped Functionals as Robustification of Value-at-Risk}
\label{sec:vr}
While we have proved that every SSD-consistent (resp.\ CSD-consistent) star-shaped functional arises as the minimum of SSD-consistent (resp.\ CSD-consistent) convex functionals, we aim to characterize the bigger family of law-invariant and star-shaped functionals in terms of Value-at-Risk.
Indeed, as explained in the previous sections, not every law-invariant shar-shaped functional can be described as the minimum of a family of law-invariant convex functionals, having the Value-at-Risk as a classical counter-example to this statement. 
See also Section~7 in~\cite{CCMTW22} and the references therein for a thorough discussion on this topic.
More specifically, we are seeking to establish representations in a similar spirit to those provided in Theorem~5 of~\cite{CCMTW22} and Proposition~A.4 of~\cite{HWWX22}. 
In \cite{CCMTW22}, law-invariant star-shaped monetary risk measures appear as a robustification of Value-at-Risk, where the generalized scenarios are represented by a penalty function dependent on the $\beta$-level of $VaR_{\beta}$. A similar representation is found in \cite{HWWX22} for cash-subadditive risk measures, even when star-shapedness is not a requirement in this case.
In the following proposition we show that this representation can be extended to the more general setting in which cash-additivity (or cash-subadditivity) and monotonicity are discarded, and $\mathcal{X}$ is a (much) more general space than $L^{\infty}$. 
The following results are presented under the assumption of normalization. It is important to note that the proof strategy remains applicable even for a \textit{non-normalized} functional.

\begin{proposition}
Let $f:\mathcal{X}\to\R\cup\{+\infty\}$ be a functional such that $f(0)=0$. 
Then, $f$ is law-invariant and star-shaped if and only if it admits the following representation:
\begin{equation}
    f(X)=\min_{Z\in\Gamma_{X}}\sup_{\beta\in(0,1)}\left\{VaR_{\beta}(X)-\gamma^X_Z(\beta)\right\}, \ X\in\mathcal{X}.
    \label{eq:VaRSS}
\end{equation}
Here, for each $X\in\mathcal{X}$, $\Gamma_X$ is a set of random variables such that $\lambda\Gamma_{X}\subseteq\Gamma_{\lambda X}$ for any $\lambda\in(0,1]$. 
In addition, for any $X,X'\in\mathcal{X}$ such that $X\sim X'$, it holds that $\Gamma_{X}=\Gamma_{X'}$. 
For each $X\in\mathcal{X}$ and $Z\in\Gamma_X$, $\gamma^X_Z:(0,1)\to\R\cup\{-\infty\}$ is an increasing function such that $\ds\max_{Z\in\Gamma_{0}}\, \gamma_Z^{0}(0^+) =0$ and for any $\lambda\in(0,1],Z\in\Gamma_{X}$ it holds that
${\gamma^{\lambda X}_{\lambda Z}}\geq\lambda\gamma^{X}_{Z}$. 

\noindent Furthermore, for each $X\in\mathcal{X}$, the set $\Gamma_X$ takes the following form:
\begin{equation*}
    \Gamma_X:=\left\{Z\in\mbox{Dom}(f): X\sim_1\alpha Z, \ \alpha\in[0,1]\right\},
\end{equation*}
and the function $\gamma^{X}_Z:(0,1)\to\R\cup\{-\infty\}$ with $Z\in\Gamma_X$ is defined by the formula:
$$\gamma^X_Z(\beta)=\alpha(VaR_{\beta}(Z)-f(Z)),$$
with that convention that if $\Gamma_X=\emptyset$ then $\gamma_Z^X=-\infty.$

\noindent With the same notation as above, $f:\mathcal{X}\to\R\cup\{+\infty\}$ is law-invariant and positively homogeneous if and only $f$ can be represented as in Equation~\eqref{eq:VaRSS}, with $\lambda\in[0,+\infty)$, $\lambda\Gamma_X=\Gamma_{\lambda X}$ and ${\gamma^{\lambda X}_{\lambda Z}}=\lambda\gamma^{X}_{Z}$. 
In addition, the set $\Gamma_X$ can be defined for all $\alpha\in[0,+\infty).$
\label{prop:VaRRo}
\end{proposition}

\begin{proof}
We prove the statement in the star-shapedness case, the positively homogeneous case follows similarly.  
We assume $X\in\mbox{Dom}(f)$, otherwise the statements are trivial.\footnote{It is worth noting that $\Gamma_X\neq\emptyset \iff X\in\mbox{Dom}(f)$. 
Indeed, if $X\in\mbox{Dom}(f)$, then we can choose $Z=X$, which clearly lies in $\Gamma_X$ with $\alpha=1$. 
Conversely, if $X\not\in\mbox{Dom}(f)$, then for any $Z\in\mathcal{X}$ such that $X\sim_1\alpha Z$, we have that $\alpha f(Z)\geq f(\alpha Z)=f(X)=+\infty$. 
Thus, $f(Z)=+\infty$, and hence $Z\not\in\Gamma_X$ since $Z\not\in\mbox{Dom}(f)$. 
This observation applies, with minor modifications, to all the results presented in Section~\ref{sec:vr}.}

We start by proving the `only if' part. 
It is worth noting that $\Gamma_X$ is well defined, given that if $X\sim_1\alpha Z$, with $\alpha\in[0,1]$ or $\alpha\in[0,+\infty)$, then $\alpha$ is unique, as we have shown in the first part of the proof of Theorem~\ref{th:mainth}.\smallskip

\textit{Min-sup representation:} Let us consider the functional
    $$X\mapsto\tilde\gamma_{Z,\beta}(X):=VaR_{\beta}(X)-\gamma^X_Z(\beta),$$
    with $\beta\in(0,1)$, and $Z\in\Gamma_X$. 
    It holds that:
    \begin{align*}
    \tilde\gamma_{Z,\beta}(X)&=VaR_{\beta}(X)-\alpha(VaR_{\beta}(Z)-f(Z)) \\
    &=VaR_{\beta}(X)-VaR_{\beta}(\alpha Z)+\alpha f(Z) \\ 
    &=\alpha f(Z)\geq f(\alpha Z) = f(X),
    \end{align*}
    where law-invariance of $VaR$ implies $VaR_{\beta}(X)-VaR_{\beta}(\alpha Z)=0$, while the inequality is due to star-shapedness of $f$ and the last equality follows from law-invariance of $f$. 
    Thus, $\ds\inf_{Z\in\Gamma_X}\sup_{\beta\in(0,1)}\tilde\gamma_{Z,\beta}(X)\geq f(X)$. 
    Observing that $X\in\Gamma_X$ and $\tilde\gamma_{X,\beta}(X)=f(X)$ for any \mbox{$\beta\in(0,1)$} we get:
    $$f(X)=\min_{Z\in\Gamma_X}\sup_{\beta\in(0,1)}\left\{VaR_{\beta}(X)-\gamma^X_Z(\beta)\right\}.$$
    
\textit{Properties of $\Gamma_X$:} It is a routine verification to prove that $\Gamma_{\lambda X}\subseteq\lambda\Gamma_{X}$ for any $\lambda\in(0,1]$, once observed that $Z\in\mbox{Dom}(f)\implies\lambda Z\in\mbox{Dom}(f)$. 
Indeed, if $Z\in\mbox{Dom}(f)$ then $f(Z)<+\infty$ and using star-shapedness of $f$ it follows that $f(\lambda Z)\leq\lambda f(Z)<+\infty$, thus $\lambda Z\in\mbox{Dom}(f)$. 
We only need to demonstrate that if $X\sim X'$, then $\Gamma_{X}=\Gamma_{X'}$. 
Indeed, if $X\sim X'$, it also follows that $X\sim_1 X'$. 
Consequently, we have:
    $$Z\in\Gamma_{X}\iff \exists \alpha\in[0,1] \ | \ X\sim_1\alpha Z \iff X'\sim_1 \alpha Z \iff Z\in\Gamma_{X'}.$$

\textit{Properties of $\gamma^{X}_Z$: }
Let us fix $\lambda\in(0,1]$. 
Due to the star-shaped property of $\rho$ and the set inclusion $\lambda\Gamma_{X}\subseteq\Gamma_{\lambda X}$ for any $\lambda \in (0,1]$, it follows that for any $Z \in \Gamma_{X}$ and $\beta \in (0,1)$:
\begin{equation*}
  \lambda\gamma^{X}_Z={\lambda}\alpha(VaR_{\beta}(Z)-f(Z))\leq\alpha(VaR_{\beta}({\lambda}Z)-f({\lambda}Z))=\gamma^{\lambda X}_{\lambda Z}.
\end{equation*}
    In addition, we have:
    $$\rho(0)=\min_{Z\in\Gamma_{0}}\sup_{\beta\in(0,1)}\{-\gamma^{0}_Z(\beta)\}=\min_{Z\in\Gamma_{0}}\{-\gamma^{0}_Z(0^+)\}=0,$$
   where the second equality follows from increasing monotonicity of $VaR_{\beta}$ w.r.t.\ $\beta$, while the last equality is due to normalization. 
   Thus, $\ds\max_{Z\in\Gamma_{0}}\gamma^{0}_Z(0^+)=0$.
\smallskip   

Now we are ready to prove the converse implication. 
We want to show that $$\rho(X)=\min_{Z\in\Gamma_{X}}\sup_{\beta\in(0,1)}\left\{VaR_{\beta}(X)-\gamma^X_Z(\beta)\right\}, \ X\in\mathcal{X},$$ is normalized, star-shaped and law-invariant.

\textit{Normalization:} We have that:
$$\rho(0)=\min_{Z\in\Gamma_0}\sup_{\beta\in(0,1)}\{-\gamma_Z^{0}(\beta)\}=-\max_{Z\in\Gamma_0}\gamma_Z^0(0^+)=0.$$
   
\textit{Star-shapedness:} Let us consider $\lambda\in(0,1]$ and $X\in\mathcal{X}$. 
We have the following inequalities:
\begin{align*}
    \rho(\lambda X)&=\min_{Z\in\Gamma_{\lambda X}}\sup_{\beta\in(0,1)}\{VaR_{\beta}(\lambda X)-\gamma_Z^{\lambda X}(\beta)\} \\
    &=\lambda\min_{Z\in\Gamma_{\lambda X}}\sup_{\beta\in(0,1)}\{VaR_{\beta}(X)-\frac{1}{\lambda}{\gamma_Z^{\lambda X}}(\beta)\} \\
    &\leq\lambda\min_{Z\in\lambda\Gamma_{X}}\sup_{\beta\in(0,1)}\{VaR_{\beta}(X)-\frac{1}{\lambda}{\gamma_Z^{\lambda X}}(\beta)\} \\
    &\leq \lambda\min_{Z\in\lambda \Gamma_{X}}\sup_{\beta\in(0,1)}\{VaR_{\beta}(X)-\gamma^X_{\frac{Z}{\lambda}}(\beta)\} \\ 
    & =\lambda\min_{Z'\in\Gamma_X}\sup_{\beta\in(0,1)}\{VaR_{\beta}(X)-\gamma^X_{Z'}(\beta)\}=\lambda\rho(X),
\end{align*}
where the inequalities follow from $\gamma_{\lambda Z}^{\lambda X}\geq\lambda \gamma^{X}_{Z}$ and $\lambda \Gamma_{X}\subseteq\Gamma_{\lambda X}.$

   \textit{Law-invariance:} This assertion follows from the law-invariance of VaR and the relation $\Gamma_{X}=\Gamma_{X'}$ as long as $X\sim X'$. 
Thus, we can express it as follows:
$$f(X)=\min_{Z\in\Gamma_X}\sup_{\beta\in(0,1)}\left\{VaR_{\beta}(X)-\gamma^X_Z(\beta)\right\}=\min_{Z\in\Gamma_{X'}}\sup_{\beta\in(0,1)}\left\{VaR_{\beta}(X')-\gamma^X_Z(\beta)\right\}=f(X').$$
\end{proof}

\begin{corollary}
    A functional $f:\mathcal{X}\to\R\cup\{+\infty\}$ is normalized, law-invariant, star-shaped and cash-additive if and only it admits the representation:
    $$f(X)=\min_{Z\in\Gamma_{X}}\sup_{\beta\in(0,1)}\{VaR_{\beta}(X)-\gamma^X_Z(\beta)\},$$
    where for each $X\in\mathcal{X}$, $\Gamma_{X}$ is a set of random variables such that $\lambda\Gamma_{X}\subseteq\Gamma_{\lambda X}$ for any $\lambda\in(0,1]$, $\Gamma_{X}=\Gamma_{X'}$ whenever $X\sim X'$ and $\Gamma_{X+m}=\Gamma_{X}$ for all $m\in\R$. 
    In addition, $\gamma_Z^X:(0,1)\to\R\cup\{-\infty\}$ is an increasing function such that $\ds\max_{Z\in\Gamma_{0}}\, \gamma_Z^{0}(0^+) =0$. 
    Moreover, for any $Z\in\Gamma_X, \ \lambda \in(0,1]$ it holds that $\lambda\gamma_{Z}^{X}=\gamma_{\lambda Z}^{\lambda X}$ and for any $m\in\R$ it follows that $\gamma_Z^{X+m}=\gamma_Z^{X}.$ 

\noindent Furthermore, for each $X\in\mathcal{X}$, the set $\Gamma_X$ takes the following form:
\begin{equation*}
    \Gamma_X:=\left\{Z\in\mbox{Dom}(f): X\sim_1\alpha Z+c, \ (\alpha,c)\in[0,1]\times\R, \ f(Z)\leq 0\right\},
\end{equation*}
with $(\alpha,c)=(0,X)$ if $X$ is constant. The function \mbox{$\gamma^{X}_Z:(0,1)\to\R\cup\{-\infty\}$} with $Z\in\Gamma_X$ is defined by the formula:
$$\gamma^X_Z(\beta)=\alpha VaR_{\beta}(Z),$$
with that convention that if $\Gamma_X=\emptyset$ then $\gamma_Z^X=-\infty.$

\noindent With the same notation as above, $f:\mathcal{X}\to\R\cup\{+\infty\}$ is law-invariant, positively homogeneous and cash-additive if and only $f$ can be represented as in Equation~\eqref{eq:VaRSS}, with $\lambda\in[0,+\infty)$ and $\lambda\Gamma_X=\Gamma_{\lambda X}$. 
In addition, the set $\Gamma_X$ can be defined for any $\alpha\in[0,+\infty)$.
\label{cor:varnonmonCA}
\end{corollary}
   \begin{proof}
       We prove the statement in the star-shapedness case, the positively homogeneous case follows similarly. 
       We commence by establishing the `only if' part. 
       It is important to note that for any $Z\in\Gamma_X$, the function $\gamma^X_Z$ is well-defined. 
       Indeed, if $X$ is non-constant and $X\sim_1\alpha Z+c$ with $(\alpha,c)\in[0,1]\times\R$, then, as demonstrated in the first part of the proof of Proposition~\ref{cor:CA}, $(\alpha,c)$ is a unique pair of values.
       Additionally, if $X$ is constant, we have $\gamma^X_Z\equiv0,$ for any $Z\in\Gamma_X=\{Z\in\mbox{Dom}(f):f(Z)\leq0\}.$ 
       Henceforth, we assume $X\in\mbox{Dom}(f)$ with $X$ non-constant, otherwise the statements are trivial.\smallskip

\textit{Min-sup representation:} Let us fix $X\in\mbox{Dom}(f)$, $\beta\in(0,1)$, and $Z\in\Gamma_X$.
    It holds that:
    \begin{align*}
    \tilde\gamma_{Z,\beta}:&=VaR_{\beta}(X)-\alpha(VaR_{\beta}(Z)-f(Z)) \\
    &=VaR_{\beta}(X)-VaR_{\beta}(\alpha Z+c)+\alpha f(Z)+c \\ &=\alpha f(Z)+c\geq f(\alpha Z+c) = f(X),
    \end{align*}
    where law-invariance and cash-additivity of $VaR$ implies $VaR_{\beta}(X)-VaR_{\beta}(\alpha Z+c)=0$, while the inequality is due to star-shapedness and cash-additivity of $f$ and the last equality follows from law-invariance of $f$. 
    Thus, $$\ds\inf_{Z\in\Gamma_X}\sup_{\beta\in(0,1)}\{VaR_{\beta}(X)-\alpha(VaR_{\beta}(Z)-f(Z))\}\geq f(X).$$ 
    Since cash-additivity implies $f(X-f(X))=0$ we have $X-f(X)\in\Gamma_X$, with \mbox{$(\alpha,c)=(1,f(X))$} and $\tilde\gamma_{X,\beta}(X)=f(X)$. 
    Thus,
    \begin{align*}
        f(X)&=\min_{Z\in\Gamma_X}\sup_{\beta\in(0,1)}\left\{VaR_{\beta}(X)-\alpha(VaR_{\beta}(Z)-f(Z))\right\} \\ 
        &=\min_{\substack{Z\in\Gamma_X \\ f(Z)=0}}\sup_{\beta\in(0,1)}\left\{VaR_{\beta}(X)-\alpha VaR_{\beta}(Z)\right\} \\
        &=\min_{\substack{Z\in\Gamma_X \\ f(Z)=0}}\sup_{\beta\in(0,1)}\left\{VaR_{\beta}(X)-\gamma_Z^X(\beta)\right\}.
    \end{align*}
    Furthermore, given that $Z\in\Gamma_{X}\implies f(Z)\leq0$ it follows that:
   \begin{align*}
       f(X)&=\min_{Z\in\Gamma_X}\sup_{\beta\in(0,1)}\left\{VaR_{\beta}(X)-\alpha(VaR_{\beta}(Z)-f(Z))\right\} \\ 
       &\leq \min_{Z\in\Gamma_X}\sup_{\beta\in(0,1)}\left\{VaR_{\beta}(X)-\alpha VaR_{\beta}(Z)\right\} \\
       &\leq \min_{\substack{Z\in\Gamma_X \\ f(Z)=0}}\sup_{\beta\in(0,1)}\left\{VaR_{\beta}(X)-\alpha VaR_{\beta}(Z)\right\}=f(X),
   \end{align*}
   hence
   \begin{equation*}
       f(X)=\min_{Z\in\Gamma_{X}}\sup_{\beta\in(0,1)}\left\{VaR_{\beta}(X)-\gamma_Z^X(\beta)\right\}.
   \end{equation*}

\textit{Properties of $\Gamma_X$:} Law-invariance of $\Gamma_X$ can be proved similarly as in Proposition~\ref{prop:VaRRo}. 
We now prove the inclusion $\lambda \Gamma_{X}\subseteq\Gamma_{\lambda X}$ for any $\lambda\in(0,1]$. 
Let $Z\in\lambda\Gamma_{X}$. 
Then there exists $Z'\in\Gamma_{X}$ such that $Z=\lambda Z'$ and $X=\alpha Z'+c$ with $f(Z')\leq 0$ and $(\alpha,c)\in[0,1]\times\R$. 
It holds that: 
$$\lambda X\sim_1\lambda (\alpha Z'+c)=\alpha Z+c',$$
where $c':=\lambda c\in\R$. 
In addition, we have $f(Z)=f(\lambda Z')\leq\lambda f(Z')\leq 0$, thus $Z\in\Gamma_{\lambda X}$. 
It is worth noting that the converse inclusion $\Gamma_{\lambda X}\subseteq\lambda\Gamma_{X}$ does not hold in general.
In addition, for $m\in\R$ and $X\in\mathcal{X}$ we have that $Z\in\Gamma_X\iff Z\in\Gamma_{X+m}$. 
Indeed, if $Z\in\Gamma_{X}$ then there exist $(\alpha,c)\in[0,1]\times\R$ such that $X\sim_1 \alpha Z+c$. 
Considering $c':=c+m$ we have $X+m\sim_1 \alpha Z +c'$, thus $Z\in\Gamma_{X+m}$.
So we showed that $\Gamma_X\subseteq\Gamma_{X+m}$; the converse inclusion can be proved analogously.
\smallskip

\textit{Properties of $\gamma^{X}_Z$:}
Let us fix $\lambda\in(0,1]$. 
By the inclusion $\lambda\Gamma_{X}\subseteq\Gamma_{\lambda X}$ for any $\lambda\in(0,1]$, it follows for any $Z\in\Gamma_{X}$ and $\beta\in(0,1)$ that:
\begin{equation*}
  \lambda\gamma^{X}_Z=\lambda \alpha VaR_{\beta}(Z)=\alpha VaR_{\beta}(\lambda Z)=\gamma^{\lambda X}_{\lambda Z},
    \end{equation*}
where last equality follows from $\lambda Z\in\lambda\Gamma_{X}\subseteq\Gamma_{\lambda X}$ with $\lambda X=\alpha \lambda Z + \lambda c$.
    The equality $\ds\max_{Z\in\Gamma_{0}}\gamma^{0}_Z(0^+)=0$ follows as in Proposition~\ref{prop:VaRRo}.
   
\noindent Finally, consider $m\in\R$. 
We have:
$$\gamma_{Z}^{X+m}=\alpha VaR(Z)=\gamma_{Z}^X,$$
where $\alpha\in[0,1]$ is the same in the representation $X\sim_1 \alpha Z+c$, whether $Z$ is regarded as an element of $\Gamma_{X}$ or as an element of $\Gamma_{X+m}$, as demonstrated earlier in the proof.
\smallskip

Now we are ready to prove the converse implication. 
We want to show that $$f(X)=\min_{Z\in\Gamma_{X}}\sup_{\beta\in(0,1)}\left\{VaR_{\beta}(X)-\gamma^X_Z(\beta)\right\}, \ X\in\mathcal{X},$$ is normalized, star-shaped, law-invariant and cash-additive. Normalization and law-invariance follows as in the proof of Proposition~\ref{prop:VaRRo}.
\smallskip

\textit{Star-shapedness:} Let us consider $\lambda\in(0,1]$. 
We have the following inequalities:
\begin{align*}
    f(\lambda X)&=\min_{Z\in\Gamma_{\lambda X}}\sup_{\beta\in(0,1)}\{VaR_{\beta}(\lambda X)-\gamma_Z^{\lambda X}(\beta)\}= \\
    &=\lambda\min_{Z\in\Gamma_{\lambda X}}\sup_{\beta\in(0,1)}\{VaR_{\beta}(X)-\frac{1}{\lambda}{\gamma_Z^{\lambda X}}(\beta)\} \\
    &= \lambda\min_{\lambda Z'\in\Gamma_{\lambda X}}\sup_{\beta\in(0,1)}\{VaR_{\beta}(X)-\frac{1}{\lambda}{\gamma_{\lambda Z'}^{\lambda X}}(\beta)\} \\ 
    &\leq\lambda\min_{\lambda Z'\in\lambda\Gamma_{X}}\sup_{\beta\in(0,1)}\{VaR_{\beta}(X)-\frac{1}{\lambda}{\gamma_{\lambda Z'}^{\lambda X}}(\beta)\}\\
    &=\lambda\min_{Z'\in\Gamma_X}\sup_{\beta\in(0,1)}\{VaR_{\beta}(X)-\gamma^X_{Z'}(\beta)\}=\lambda f(X),
\end{align*}
where the inequality follows from $\lambda \Gamma_{X}\subseteq \Gamma_{\lambda X},$ while $\gamma_{\lambda Z}^{\lambda X}=\lambda\gamma^{X}_{Z}$ for any $Z\in\Gamma_{X}$ leads to the last equality.
\smallskip

\textit{Cash-additivity:} Let $m\in\R$. 
We have:
\begin{align*}
f(X+m) &= \min_{Z\in\Gamma_{X+m}}\sup_{\beta\in(0,1)}\left\{VaR_{\beta}(X+m)-\gamma^{X+m}_Z(\beta)\right\} \\ 
&= \min_{Z\in\Gamma_{X+m}}\sup_{\beta\in(0,1)}\left\{VaR_{\beta}(X)-\gamma^{X+m}_Z(\beta)\right\}+m \\
&= \min_{Z\in\Gamma_{X}}\sup_{\beta\in(0,1)}\left\{VaR_{\beta}(X)-\gamma^X_Z(\beta)\right\}+m = f(X)+m,
\end{align*}
The last equality follows from $\Gamma_{X+m}=\Gamma_{X}$ and $\gamma_Z^{X+m}=\gamma_{Z}^X$.
\end{proof}
In the subsequent lemma, we introduce a method for representing law-invariant risk measures using Value-at-Risk as a key component. 
Subsequently, we tailor our findings to derive a novel representation that adheres to the axioms of star-shapedness and cash-additivity, as expounded in Proposition~\ref{cor:CAVaR1}.

\begin{lemma}
    Let $\rho:\mathcal{X}\to\R\cup\{+\infty\}$ be a normalized, law-invariant risk measure. 
    Then it admits the representation:
    $$\rho(X)=\min_{
        Z\in\Gamma_{X}}\rho\left(\sup_{\beta\in(0,1)}(VaR_{\beta}(X)-\gamma^X_Z(\beta))\right).$$
        Here,
\begin{equation*}
    \Gamma_X:=\left\{Z\in\mbox{Dom}(\rho): \exists Y\succeq_1 X \mbox{ s.t. } Y\sim_1\alpha Z, \ \alpha\in[0,1]\right\},
\end{equation*}
and for each fixed $X\in\mathcal{X}$ and $Z\in\Gamma_X$, we define $$\mathcal{A}_Z^{X}:=\{\alpha\in[0,1]: \exists Y\succeq_1 X \mbox{ s.t. } Y\sim_1\alpha Z\},$$ and the real function $\gamma^X_Z:(0,1)\to\R\cup\{-\infty\}:$
$$\gamma^X_Z(\beta):=\bar{\alpha}(VaR_{\beta}(Z)-Z),$$
where $\bar{\alpha}=\sup\mathcal{A}_Z^{X}$, with that convention that if $\Gamma_X=\emptyset$ then $\gamma_Z^X=-\infty.$
        \label{lem:monSSVaR}
\end{lemma}
\begin{proof}
We observe that $\gamma_Z^X$ is well-defined for any $X\in\mathcal{X}$ and $Z\in\mbox{Dom}(\rho)$, as we consider $\bar\alpha=\sup\mathcal{A}_Z^X$, which is clearly unique. Proceeding analogously as in the proof of Proposition~A.4 in \cite{HWWX22}, it can be verified that any law-invariant risk measure can be represented as:
    $$\rho(X)=\min_{\substack{W\in\mathcal{X} \\ W\succeq_1 X}}\rho(W+\sup_{\beta\in(0,1)}(VaR_{\beta}(X)-VaR_{\beta}(W)).$$
    We fix $X\in\mbox{Dom}(\rho)$, otherwise the statements are trivial. 
    If $Z\in\Gamma_{X}$, then for any $\alpha\in\mathcal{A}_Z^{X}$ it holds that $\alpha Z\succeq_1 X$, implying that also $\bar{Y}:=\bar\alpha Z$ verifies $\bar\alpha Z\succeq_1 X$. 
    Indeed, by the properties of the supremum, there exists a sequence $(\alpha_n)_{n\in\mathbb{N}}\subseteq\mathcal{A}_Z^{X}$ such that $\alpha_n\to\bar\alpha$. 
    Thus, for any $n\in\mathbb{N}$ it results that $Y_n\succeq_1 X$ and $Y_n\sim_1\alpha_n Z$. 
    Letting $n\to\infty$, we have that $\bar{Y}:=\bar\alpha Z$ clearly satisfies $\bar{Y}\sim_1\bar\alpha Z$ and $\bar{Y}\succeq_1 X$, so the supremum is indeed a maximum. 
    Hence, for any $Z\in\Gamma_X$ we have
    \begin{equation}
    \rho(X)\leq \rho(\bar\alpha Z+\sup_{\beta\in(0,1)}(VaR_{\beta}(X)-\bar\alpha VaR_{\beta}(Z)).
    \label{eq:MinZ}
    \end{equation}
    Taking the minimum over $Z\in\Gamma_X$ on both members of Equation~\eqref{eq:MinZ} we have:
    \begin{align*}
    \rho(X)&=\min_{\substack{W\in\mathcal{X} \\ W\succeq_1 X}}\rho(W+\sup_{\beta\in(0,1)}(VaR_{\beta}(X)-VaR_{\beta}(W)) \\
    &\leq\min_{Z\in\Gamma_{X}}\rho(\bar\alpha Z+\sup_{\beta\in(0,1)}(VaR_{\beta}(X)-\bar\alpha VaR_{\beta}(Z)) \\
    &=\min_{Z\in\Gamma_{X}}\rho(\sup_{\beta\in(0,1)}(VaR_{\beta}(X)-\gamma^X_Z(\beta))\leq\rho(X),
    \end{align*}
    where the last inequality follows by taking $Z=X$, given that $\bar{\alpha}=1$ attains the supremum of the set $\mathcal{A}^X_X$. 
\end{proof}

The following proposition demonstrates that our results genuinely extend those obtained in Theorem~5 of~\cite{CCMTW22}. 
Not only do we broaden the setting to encompass the general space $\mathcal{X}$, but we also establish that the minimum can be taken over a set strictly smaller than the acceptance set of $\rho$, denoted as $\mathcal{B}{\rho}:=\{Z\in\mathcal{X}:\rho(Z)\leq0\}$.
More specifically, the corollary underscores that the minimum can be computed over a set $\Gamma_X\subseteq \mathcal{B}{\rho}$, which is contingent on the choice of $X\in\mathcal{X}$. 
\begin{proposition}
\label{cor:CAVaR1}
    A functional $\rho:\mathcal{X}\to\R\cup\{+\infty\}$ is a normalized, law-invariant, star-shaped and cash-additive risk measure if and only it admits the representation:
    $$\rho(X)=\min_{Z\in\Gamma_{X}}\sup_{\beta\in(0,1)}\{VaR_{\beta}(X)-\gamma^X_Z(\beta)\},$$
    where for each $X\in\mathcal{X}$, $\Gamma_{X}$ is a set of random variables such that $\lambda\Gamma_{X}\subseteq\Gamma_{\lambda X}$ for any $\lambda\in(0,1]$, and $\Gamma_{X+m}=\Gamma_{X}$ for all $m\in\R$. 
    In addition, $\gamma_Z^X:(0,1)\to\R\cup\{-\infty\}$ is an increasing function such that $\ds\max_{Z\in\Gamma_{0}}\, \gamma_Z^{0}(0^+) =0$. 
    Moreover, for any $Z\in\Gamma_X, \ \lambda \in(0,1]$ it holds that $\lambda\gamma_{Z}^{X}=\gamma_{\lambda Z}^{\lambda X}$ and for any $m\in\R$ it results that $\gamma_Z^{X+m}=\gamma_Z^{X}.$ 

\noindent Furthermore, for each $X\in\mathcal{X}$, the set $\Gamma_X$ takes the following form:
\begin{equation*}
    \Gamma_X:=\left\{Z\in\mbox{Dom}(\rho): \exists Y\succeq_1X, Y\sim_1\alpha Z+c, \ (\alpha,c)\in[0,1]\times\R, \ \rho(Z)\leq 0\right\},
\end{equation*}
and the function $\gamma^{X}_Z:(0,1)\to\R\cup\{-\infty\}$ is defined by the formula:
$$\gamma^X_Z(\beta):=\bar\alpha VaR_{\beta}(Z), \ \ Z\in\Gamma_X,$$
with that convention that if $\Gamma_X=\emptyset$ then $\gamma_Z^X=-\infty.$ 
Here, $\bar{\alpha}:=\sup\mathcal{A}_Z^{X}$ where 
$$\mathcal{A}_Z^{X}:=\{\alpha\in[0,1]: \exists Y\succeq_1 X, c\in\R \mbox{ s.t. } \ Y\sim_1\alpha Z+c\}.$$
\end{proposition}
\begin{proof}
We start by proving the `only if' implication. 
We assume $X\in\mbox{Dom}(\rho)$, otherwise the statements are trivial. 
For every $X\in\mathcal{X}$ and $Z\in\Gamma_X$, the function $\gamma^X_Z$ is well-defined, similarly to what occurs in Lemma \ref{lem:monSSVaR}.    Furthermore, we observe that $\sup\mathcal{A}_Z^X=\max\mathcal{A}_Z^X$, when $\mathcal{A}_Z^X\neq\emptyset$. 
    It is easy to verify that if $Z\in\Gamma_{X}$ is non-constant and $Y$ is non-constant, there exists a unique $(\alpha,c)\in[0,1]\times\R$ such that $Y\sim_1 \alpha Z+c$ (see the proof of Proposition~\ref{cor:CA}).
Thus, if $\mathcal{A}_Z^X$ is non-empty, it consists of a single element, and the supremum is reached. When $Y$ is constant and $Z$ is non-constant, the only admissible choice is $(\alpha,c)=(0,Y)$, and the supremum is attained. Finally, if $Z$ is constant, then $Y$ must also be constant, and we can choose $(\bar\alpha,\bar c)=(1,Y-Z)$. 
In this case, $\bar\alpha=1$ is an element of $\mathcal{A}_Z^X$ which reaches the supremum, as $\mathcal{A}_Z^X\subseteq[0,1]$. \smallskip

\textit{Min-sup representation:} First, we note that fixing $X\in\mbox{Dom}(\rho)$ and $Z\in\Gamma_X$ and defining $$X\mapsto \tilde{\gamma}_{Z}(X):=\sup_{\beta\in(0,1)}\{(VaR_{\beta}(X)-VaR_{\beta}(\bar\alpha Z))+\bar\alpha\rho(Z))\},$$ with $\beta\in(0,1)$, it holds that:
\begin{equation*}
    \tilde{\gamma}_Z(X)\geq\rho(\bar\alpha Z+\sup_{\beta\in(0,1)}\{VaR_{\beta}(X)-VaR_{\beta}(\bar\alpha Z)\}).
\end{equation*}
Here, the inequality is due to star-shapedness and cash-additivity of $\rho$.
We have $$\rho(\bar\alpha Z+\sup_{\beta\in(0,1)}\{VaR_{\beta}(X)-VaR_{\beta}(\bar\alpha Z)\})=\rho(\bar\alpha Z+\bar c+\sup_{\beta\in(0,1)}\{VaR_{\beta}(X)-VaR_{\beta}(\bar\alpha Z+\bar c)\}),$$ 
where $\bar c\in\R$ is the element such that $Y\sim_1 \bar\alpha Z+\bar c$. 
Such $\bar c\in\R$ exists since $\bar\alpha\in\mathcal{A}_Z^X$, as we proved in the first part of the proof. 
Thus, $\bar\alpha Z+\bar c\succeq_1 X$, resulting in:
\begin{align*}
\tilde{\gamma}_{Z}(X) & \geq \rho\left(\bar\alpha Z+\bar c+\sup_{\beta\in(0,1)}{VaR_{\beta}(X)-VaR_{\beta}(\bar\alpha Z+\bar c)}\right) \\
& \geq \min_{\substack{W\in\mathcal{X} \\ W\succeq_1 X}}\rho\left(W+\sup_{\beta\in(0,1)}(VaR_{\beta}(X)-VaR_{\beta}(W))\right) = \rho(X),
\end{align*}
where the last equality follows from the proof of Lemma~\ref{lem:monSSVaR}. 
Hence, $\ds\min_{Z\in\Gamma_X}\tilde{\gamma}_Z(X)\geq \rho(X)$. 
Since cash-additivity yields $\rho(X-\rho(X))=0$ it holds that $\tilde{Z}:=X-\rho(X)\in\Gamma_X$ with $(\bar\alpha,\bar c)=(1,\rho(X))$. 
In particular, we have $\tilde{\gamma}_{\tilde Z}(X)=\rho(X)$, leading to:
  \begin{align*}
        \rho(X)&=\min_{Z\in\Gamma_X}\sup_{\beta\in(0,1)}\left\{VaR_{\beta}(X)-\bar\alpha(VaR_{\beta}(Z)-\rho(Z))\right\} \\ 
        &=\min_{\substack{Z\in\Gamma_X \\ \rho(Z)=0}}\sup_{\beta\in(0,1)}\left\{VaR_{\beta}(X)-\bar\alpha VaR_{\beta}(Z)\right\} \\
        &=\min_{\substack{Z\in\Gamma_X \\ \rho(Z)=0}}\sup_{\beta\in(0,1)}\left\{VaR_{\beta}(X)-\gamma_Z^X(\beta)\right\}.
    \end{align*}
    Furthermore, given that $Z\in\Gamma_{X}\implies \rho(Z)\leq0$ it follows that:
   \begin{align*}
       \rho(X)&=\min_{Z\in\Gamma_X}\sup_{\beta\in(0,1)}\left\{VaR_{\beta}(X)-\bar\alpha(VaR_{\beta}(Z)-\rho(Z))\right\} \\ 
       &\leq \min_{Z\in\Gamma_X}\sup_{\beta\in(0,1)}\left\{VaR_{\beta}(X)-\bar\alpha VaR_{\beta}(Z)\right\} \\
       &\leq \min_{\substack{Z\in\Gamma_X \\ \rho(Z)=0}}\sup_{\beta\in(0,1)}\left\{VaR_{\beta}(X)-\bar\alpha VaR_{\beta}(Z)\right\}=\rho(X),
   \end{align*}
   hence
   \begin{equation*}
       \rho(Z)=\min_{Z\in\Gamma_{X}}\sup_{\beta\in(0,1)}\left\{VaR_{\beta}(X)-\gamma_Z^X(\beta)\right\}.
   \end{equation*}

\noindent The properties of $\Gamma_X$ properties of $\gamma_Z^{X}$  can be verified as in the proof of Corollary~\ref{cor:varnonmonCA}. \smallskip

Now we are ready to prove the converse implication. 
We want to show that $$\rho(X)=\min_{Z\in\Gamma_{X}}\sup_{\beta\in(0,1)}\left\{VaR_{\beta}(X)-\gamma^X_Z(\beta)\right\}, \ X\in\mathcal{X},$$ is monotone, normalized, star-shaped, law-invariant and cash-additive. 
Normalization and star-shapedness can be established following the proof of Proposition~\ref{prop:VaRRo}, while the proof of cash-additivity closely mirrors the proof provided in Corollary~\ref{cor:varnonmonCA}.
\smallskip

\textit{Monotonicity:} This property is obvious once observed that $VaR_{\beta}$ is monotone and $\Gamma_{X_1}\subseteq\Gamma_{X_2}$ as soon as $X_1\geq X_2$. \smallskip

\textit{Law-invariance:} Let us consider $X,X'\in\mathcal{X}$ such that $X\sim X'$. 
We need to verify that $\Gamma_{X}=\Gamma_{X'}$. 
Given that $X\sim X'$, both $X\succeq_1 X'$ and $X'\succeq_1 X$ apply, resulting in $\Gamma_{X}\subseteq\Gamma_{X'}$ and $\Gamma_{X'}\subseteq\Gamma_{X}$, which in turns lead to $\Gamma_{X}=\Gamma_{X'}$. This equality, along with the law-invariance of $VaR_{\beta}$, implies that $\rho$ is law-invariant. 
\end{proof}

\section{Illustrative Examples}\label{sec:ex}
While the prime motivation for this paper comes directly from \cite{BLR18,BKMS21,LR22,LRZ23,ABL23}, whence the aim of establishing novel characterization results for law-invariant return and star-shaped risk measures, we show in this section that such risk measures may also arise naturally from more classical settings. 
Indeed, we provide three examples to illustrate the inherent star-shaped nature of certain risk measures, also when cash-additivity or cash-subadditivity are not preserved.
Specifically, the examples show that during the recent extended period of negative interest rates induced by central banks' monetary policies, investors may naturally comply with non-cash-(sub)additive risk measures that remain star-shaped or even positively homogeneous. 
\begin{example}
As shown in \cite{ELKR09}, in the absence of a zero coupon bond, the presence of ambiguity with respect to interest rates naturally leads to a relaxation of the axiom of cash-additivity, with cash-subadditivity being assumed instead. 
Over the last decade, the EONIA index, which tracks unsecured lending transactions in the interbank market, has assumed negative values. 
Consider a bank, or more generally, a financial institution, assessing the risk of an asset $X_T$ at the present time $t=0$, where $T>0$ denotes the asset's time to maturity. 
Suppose the institution evaluates the risk of $X_T$ using a spot risk measure $\rho_0$ defined on the discounted price $D_TX_T$, where $D_T$ represents a discount factor (for more details, see Section~2.4 in \cite{ELKR09}). 
In the context of negative interest rates, the discount factor $D_T$ can take values greater than 1. 
Therefore, if the interest rate is subject to ambiguity, fluctuating between two constants $0 \leq D_b \leq D_u \leq C$, with $C > 1$, an ambiguity-adverse investor may select the risk measure as follows:
    $$\tilde{\rho}(X_T)=\sup_{D_T\in\mathcal{X}}\{\rho_0(D_TX_T): D_b\leq D_T\leq D_u\}.$$
    Even when $\rho_0$ is a monetary risk measure, $\tilde{\rho}$ can be non-cash-(sub)additive, given that $D_T$ can take values greater than $1$.
    
    A classic industry measure of risk is the Value-at-Risk. 
    Therefore, we can consider $\rho_0 = VaR_{\beta}$. 
    In this setting, the resulting risk measure is given by:
    $$\tilde{\rho}(X_T)=\sup_{D_T\in\mathcal{X}}\{VaR_{\beta}(D_TX_T): D_b\leq D_T\leq D_u\}.$$ 
    This risk measure is neither convex nor cash-(sub)additive, but it is positively homogeneous (thus star-shaped) and law-invariant, inheriting these properties from $VaR_{\beta}$. 
    In particular, the thesis of Proposition~\ref{prop:VaRRo} concerning positively homogeneous functionals can be applied to represent $\tilde{\rho}$ as in Equation~\eqref{eq:VaRSS}. 
    \label{ex:VaR1}
\end{example}

\begin{example}
In the same context as in Example~\ref{ex:VaR1}, we consider a generalization of $VaR_{\beta}$ that is closely connected to the concept of $\Lambda$-VaR, as shown in Theorem~3.1 of \cite{HWWX22}.  
For a thorough discussion of $\Lambda$-VaR risk measures, we refer to \cite{FMP13}.
Let us fix $x\in\mathbb{R}$ and define:
\begin{align*}
    \tilde{\rho}(X_T)&=\sup_{D_T\in\mathcal{X}}\{VaR_{\beta}(D_TX_T)\vee x: D_b\leq D_T\leq D_u\} \\ &=\sup_{D_T\in\mathcal{X}}\{VaR_{\beta}(D_TX_T): D_b\leq D_T\leq D_u\}\vee x.
\end{align*}
Since the maximum operation preserves star-shapedness, and both $VaR_{\beta}$ and the trivial risk measure $\rho(X)=x$ for any $X\in\mathcal{X}$ are star-shaped, the resulting risk measure $\tilde{\rho}$ is law-invariant and (genuinely) star-shaped. 
However, in this case, positive homogeneity does not hold in general. 
Once again, $\tilde{\rho}$ is neither convex nor cash-(sub)additive. 
The parameter $x$ can be interpreted as follows: it represents a barrier below which the risk evaluation corresponding to the asset $X_T$ cannot fall, possibly due to market frictions or other constraints.
Consequently, the investor is required to retain the larger of two amounts --- the assessment of the risk linked to the asset $X_T$ as determined by the Value-at-Risk and the minimum threshold $x \in \mathbb{R}$. 
The fixed amount $x$ remains constant and is unrelated to the value of $X_T$, but instead may be influenced by external factors. 
In this case, it is also possible to derive a representation of $\tilde{\rho}$ akin to the one presented in Equation~\eqref{eq:VaRSS}. 
It should be noted that $\tilde{\rho}$ is not normalized. 
However, the representation results remain valid, as stated at the outset of Section~\ref{sec:vr}.
\end{example}

\begin{example}
    Consider a firm that wishes to establish an insurance contract to cover the risk associated with an asset $X_T$. 
    Let $A$ represent the set of insurance companies accessible to the firm. 
    Suppose each insurer has its own spot monetary, convex, and normalized risk measure $\rho_0^a$ with $a\in A$ to assess the risk of $X_T$. 
    Furthermore, suppose the cost of each contract is equal to the risk assessment made by the insurer through its risk measure. 
    Thus, for each asset $X_T$, the firm will choose to pay the minimum amount:
    $$\rho(X_T):=\min_{a\in A}\rho_0^a(D_TX_T).$$ 
    Here, the resulting risk measure is star-shaped but not convex in general. 
    Furthermore, as observed in the examples above, $\rho$ lacks cash-(sub)additivity. 
    If the risk measures employed by the insurance companies are SSD-consistent, such as in the case of Expected Shortfall, entropic risk measures, and risk measures generated from power or exponential utilities, the resulting risk measure $\rho$ inherits this property. 
    Under this circumstance, if $\mathcal{X}$ is a rearrangement-invariant space, we can represent $\rho$ as in Equation~\eqref{eq:SDD_RS}.
\end{example}
These three examples underscore the significance of establishing a comprehensive framework for law-invariant (or SSD-consistent) star-shaped risk measures, which remains applicable also when cash-(sub)additivity is not preserved.

\end{document}